\documentclass[envcountsame,envcountsect,orivec]{llncs}

\usepackage{amssymb,amsmath,stmaryrd}
\usepackage{mathpartir,xspace}
\usepackage{graphicx}
\usepackage{mathtools}

\usepackage{url}

\newcommand{\Sl}{\ensuremath{\mathcal{S}}}
\newcommand{\Bl}{\ensuremath{\mathcal{B}}}

\newcommand{\Wl}{\ensuremath{\mathcal{W}}}
\newcommand{\Fl}{\ensuremath{\mathcal{F}}} 
\newcommand{\Dl}{\ensuremath{\mathcal{D}}}
\newcommand{\Rl}{\ensuremath{\mathcal{R}}}

\newcommand{\Wr}{\ensuremath{\mathfrak{W}}}

\newcommand{\rr}{\ensuremath{\mathfrak{r}}}

\newcommand{\NN}{\ensuremath{\mathbb{N}}}

\newcommand{\QQ}{\ensuremath{\mathbb{Q}}}

\newcommand{\RAND}{\ensuremath{\syn{rand}\,}}

\newcommand{\lang}{\ensuremath{\mathbf{F}^{\mu,\oplus}}}
\newcommand{\syn}[1]{\textnormal{\texttt{#1}}}
\newcommand{\bnfeq}{\ensuremath{::=}}
\newcommand{\hmid}{\ensuremath{\;|\;}}
\newcommand{\BOOL}{\ensuremath{\mathbf{2}}}
\newcommand{\TRUE}{\ensuremath{\syn{true}}}
\newcommand{\FALSE}{\ensuremath{\syn{false}}}
\newcommand{\ONE}{\ensuremath{\mathbf{1}}}
\newcommand{\plus}{\ensuremath{\mathop{+}}}
\newcommand{\NAT}{\syn{nat}}
\newcommand{\ALL}[2]{\ensuremath{\forall#1.#2}}
\newcommand{\EX}[2]{\ensuremath{\exists#1.#2}}
\newcommand{\REC}[2]{\ensuremath{\mu#1.#2}}

\newcommand{\INL}{\ensuremath{\syn{inl}\,}}
\newcommand{\INR}{\ensuremath{\syn{inr}\,}}
\newcommand{\FOLD}{\ensuremath{\syn{fold}\,}}
\newcommand{\UNFOLD}{\ensuremath{\syn{unfold}\,}}
\newcommand{\PROJ}[1]{\ensuremath{\syn{proj}_{#1}}}

\newcommand{\unitval}{\langle\rangle}

\newcommand{\PAIR}[2]{\ensuremath{\langle#1,#2\rangle}}
\newcommand{\FUN}[2]{\ensuremath{\lambda#1.#2}}
\newcommand{\TFUN}[1]{\ensuremath{\Lambda .#1}}
\newcommand{\TAPP}[1]{\ensuremath{#1[]}}
\newcommand{\PACK}[1]{\ensuremath{\syn{pack}\,#1}}
\newcommand{\UNPACK}[3]{\ensuremath{\syn{unpack}~#1~\syn{as}~#2~\syn{in}~#3}}
\newcommand{\MATCH}[5]{\ensuremath{\syn{match}\left(#1, #2.#3, #4.#5\right)}}

\newcommand{\IFZERO}[3]{\ensuremath{\syn{if}_1~#1~\syn{then}~#2~\syn{else}~#3}}
\newcommand{\PRED}{\ensuremath{\syn{P}\,}}
\newcommand{\SUCC}{\ensuremath{\syn{S}\,}}
\newcommand{\IFT}[3]{\ensuremath{\syn{if}~#1~\syn{then}~#2~\syn{else}~#3}}

\newcommand{\FIX}{\syn{fix}\xspace}
\newcommand{\REFTY}{\ensuremath{\syn{ref}\xspace}}
\newcommand{\LOOKUP}[1]{\ensuremath{{!{#1}}}}
\newcommand{\LOC}{\ensuremath{\syn{Loc}}}

\newcommand{\REF}{\ensuremath{\syn{ref}\,}}
\newcommand{\sat}[1]{\ensuremath{\left\lfloor #1 \right\rfloor}}
\newcommand{\ASSIGN}[2]{\ensuremath{{#1 := #2}}}
\newcommand{\basicstepsto}[1]{\ensuremath{\mathrel{\overset{#1}{\longmapsto}}}}
\newcommand{\stepsto}[1]{\ensuremath{\mathrel{\overset{#1}{\leadsto}}}}
\newcommand{\stepstopure}{\ensuremath{\mathrel{\overset{\text{cf}}{\Longrightarrow}}}}
\newcommand{\stepstozero}{\ensuremath{\mathrel{\overset{\text{uff}}{\Longrightarrow}}}}
\newcommand{\stepstozeropure}{\ensuremath{\mathrel{\overset{\text{cuff}}{\Longrightarrow}}}}

\newcommand{\conf}[2]{\ensuremath{\left\langle #1, #2 \right\rangle}}

\newcommand{\empeval}{\ensuremath{\mathop{-}}}

\newcommand{\subst}[2]{\ensuremath{[{#2}{/}{#1}]}}
\newcommand{\emp}{\ensuremath{\varnothing}}

\newcommand{\ftv}[1]{\ensuremath{\mathit{ftv}(#1)}}

\newcommand{\dom}[1]{\ensuremath{\mathrm{dom}\left(#1\right)}}
\newcommand{\extend}[3]{\ensuremath{{#1\left[#2\mapsto #3\right]}}}
\newcommand{\ecomp}[2]{\ensuremath{#1 \circ #2}}

\newcommand{\hastype}[4]{\ensuremath{{#1 \hmid #2 \vdash #3 : #4}}}

\newcommand{\dists}{\ensuremath{\mathbf{Dist}}}

\newcommand{\redpath}[3]{\ensuremath{#1 : #2 \leadsto^* #3}}
\newcommand{\last}[1]{\ensuremath{\ell\left(#1\right)}}
\newcommand{\len}[1]{\ensuremath{\mathbf{len}\left(#1\right)}}
\newcommand{\redsname}{\ensuremath{\mathbf{Red}}}
\newcommand{\reds}[1]{\ensuremath{\redsname\left(#1\right)}}
\newcommand{\paths}{\ensuremath{\mathfrak{R}}}
\newcommand{\weight}[1]{\ensuremath{\Wr\left(#1\right)}}
\newcommand{\interval}{\ensuremath{\mathcal{I}}}
\newcommand{\PR}[1]{\ensuremath{\mathfrak{P}^\Downarrow\left(#1\right)}}
\newcommand{\PRk}[2]{\ensuremath{\mathfrak{P}^\Downarrow_{#1}\left(#2\right)}}

\newcommand{\DPRk}[1]{\ensuremath{\mathbf{\mathfrak{P}\rr}^k}}
\newcommand{\msn}[1]{\mathbf{#1}}

\newcommand{\VRel}[2]{\msn{VRel}\left(#1, #2\right)}
\newcommand{\WRel}[2]{\msn{WRel}\left(#1, #2\right)}
\newcommand{\SRel}[2]{\msn{SRel}\left(#1, #2\right)}
\newcommand{\ERel}[2]{\msn{TRel}\left(#1, #2\right)} 

\newcommand{\VRelD}[1]{\msn{VRel}\left(#1\right)}
\newcommand{\Type}{\mathfrak{T}\xspace}

\newcommand{\Exp}[1]{\msn{Tm}\left( #1 \right)}
\newcommand{\Val}[1]{\msn{Val}\left( #1 \right)}
\newcommand{\Stk}[1]{\msn{Stk}\left( #1 \right)}
\newcommand{\Exps}{\msn{Tm}}
\newcommand{\Vals}{\msn{Val}}
\newcommand{\Ectxt}{\msn{Stk}\xspace}
\newcommand{\Ectxts}{\msn{Stk}\xspace}
\newcommand{\Stks}{\msn{Stk}\xspace}

\newcommand{\Subst}{\msn{Subst}\xspace}

\newcommand{\R}{\ensuremath{\mathcal{R}}\xspace} 

\newcommand{\ctxapproxrel}{\ensuremath{\mathrel{\lesssim}^{\textit{ctx}}}}
\newcommand{\ctxapprox}[5]{\ensuremath{#1 \hmid #2 \vdash #3
    \ctxapproxrel #4 : #5}}
\newcommand{\ctxequivrel}{\ensuremath{\mathrel{=}^{\textit{ctx}}}}
\newcommand{\ctxequiv}[5]{\ensuremath{#1 \hmid #2 \vdash #3
    \ctxequivrel #4 : #5}}
\newcommand{\logapproxrel}{\ensuremath{\mathrel{\lesssim}^{\textit{log}}}}
\newcommand{\ciuapproxrel}{\ensuremath{\mathrel{\lesssim}^{\textit{CIU}}}}
\newcommand{\ciuapprox}[5]{\ensuremath{#1 \hmid #2 \vdash #3
    \ciuapproxrel #4 : #5}}

\newcommand{\logapprox}[5]{\ensuremath{#1 \hmid #2 \vdash #3
    \logapproxrel #4 : #5}}
\newcommand{\logequivrel}{\ensuremath{\mathrel{\cong}^{log}}}
\newcommand{\logequiv}[5]{\ensuremath{#1 \hmid #2 \vdash #3
    \logequivrel #4 : #5}}

\newcommand{\TT}[1]{\ensuremath{{#1}^{\top\!\!\top}}}
\newcommand{\T}[1]{\ensuremath{{#1}^{\top}}}
\newcommand{\Te}[1]{\ensuremath{{#1}^{\bot}}}

\newcommand{\pow}[1]{\ensuremath{\mathcal{P}\left(#1\right)}}

\renewcommand{\implies}{\to}

\newcommand{\isetsep}{\ensuremath{\ \big|\ }}
\renewcommand{\phi}{\varphi}

\newcommand{\comp}{\ensuremath{\circ}}

\usepackage{tikz-cd}

\newcommand{\den}[1]{
 \left\llbracket #1
 \right\rrbracket}

\begin{document}

\frontmatter          
\pagestyle{headings}  

\mainmatter              
\title{Step-Indexed Logical Relations for Probability}
\author{Ale\v{s} Bizjak \and Lars Birkedal}
\authorrunning{Ale\v{s} Bizjak and Lars Birkedal} 
\institute{Aarhus University\\
  \email{\{abizjak,birkedal\}@cs.au.dk}}

\maketitle              

\begin{abstract}

  It is well-known that constructing models of higher-order probabilistic programming
  languages is challenging.  We show how to construct step-indexed logical relations for a
  probabilistic extension of a higher-order programming language with impredicative
  polymorphism and recursive types. We show that the resulting logical relation is sound
  and complete with respect to the contextual preorder and, moreover, that it is
  convenient for reasoning about concrete program equivalences. Finally, we extend the
  language with dynamically allocated first-order references and show how to extend the
  logical relation to this language. We show that the resulting relation remains useful
  for reasoning about examples involving both state and probabilistic choice.
\end{abstract}

\section{Introduction}
\label{sec:introduction}

It is well known that it is challenging to develop techniques for
reasoning about programs written in probabilistic higher-order
programming languages.  A probabilistic program evaluates to a
distribution of values, as opposed to a set of values in the case of
nondeterminism or a single value in the case of deterministic
computation. Probability distributions form a monad. This observation
has been used as a basis for several denotational domain-theoretic
models of probabilistic languages and also as a guide for designing
probabilistic languages with monadic
types~\cite{Jones:powerdomain-evaluations,SahebDjahromi:powerdomains,Ramsey:stochastic-lambda-calc}.
Game semantics has also been used to give models of
probabilistic programming
languages~\cite{Danos:prob-game-semantics,Ehrhard:computational-meaning-probch}
and a fully abstract model using coherence spaces for PCF with probabilistic choice
was recently
presented~\cite{Ehrhard:prob-coherence-fully-abstract}.

The majority of models of probabilistic programming languages have
been developed using denotational semantics.  However, Johann
et.al.~\cite{johann-et-al:operational-metatheory} developed
operationally-based logical relations for a polymorphic programming
language with effects. Two of the effects they considered were
probabilistic choice and \emph{global} ground store.  However, as
pointed out by the authors~\cite{johann-et-al:operational-metatheory},
extending their construction to local store and, in particular,
higher-order local store, is likely to be problematic. 
Recently, operationally-based bisimulation techniques have been
extended to probabilistic extensions of
PCF~\cite{Cruiblle:bisim-prob-CBV,DalLago:bisim-prob-CBN}.
The operational semantics of probabilistic higher-order
programming languages has been investigated
in~\cite{dallago:prob-operational-sem}.

Step-indexed logical
relations~\cite{Ahmed:step-indexed-quantified,Appel:2001:indexed-model}
have proved to be a successful method for proving contextual
approximation and equivalence for programming languages with a wide
range of features, including computational effects.

In this paper we show how to extend the method of step-indexed logical
relations to reason about contextual approximation and equivalence of
probabilistic higher-order programs.  To define the logical relation
we employ
biorthogonality~\cite{Pitts00parametricpolymorphism,Pitts:step-indexed-biorthogonality}
and step-indexing. Biorthogonality is used to ensure
completeness of the logical relation with respect to contextual
equivalence, but it also makes it possible to keep the value relations
simple, see Fig.~\ref{fig:logical-relation}.
Moreover, the definition using biorthogonality makes it possible to
``externalize'' the reasoning in many cases when proving example
equivalences. By this we mean that the reasoning reduces to
algebraic manipulations of probabilities. This way, the quantitative
aspects do not complicate the reasoning much, compared to the
usual reasoning with step-indexed logical relations.
To define the
biorthogonal lifting we use two notions of observation;
the termination probability and its stratified version approximating it. We define these
and prove the required properties in Section~\ref{sec:term-relat}.

We develop our step-indexed logical relations for the
call-by-value language $\lang$. This is system $\mathbf{F}$ with recursive
types, extended with a single probabilistic choice primitive
$\RAND$. The primitive $\RAND$ takes a
natural number $n$ and reduces with uniform probability to one of $1,
2, \ldots, n$. Thus $\RAND n$ represents the uniform probability
distribution on the set $\{1, 2, \ldots, n\}$. We choose to add
$\RAND$ instead of just a single coin flip primitive to make the
examples easier to write.

To show that the model is useful we use it to prove some example equivalences in
Section~\ref{sec:examples}. We show two examples based on parametricity.
In the first example, we characterize elements of the
universal type $\ALL{\alpha}{\alpha\to\alpha}$. In a deterministic language, and even in a
language with nondeterministic choice, the only interesting element of this type is the
identity function. However, since in a probabilistic language we not only observe the
end result, but also the likelihood with which it is returned, it turns out that there are
many more elements. Concretely, we show that the elements of the type
$\ALL{\alpha}{\alpha\to\alpha}$ that are of the form $\Lambda\alpha . {\FUN{x}{e}}$,
correspond precisely to \emph{left-computable} real numbers in the interval $[0,1]$.  In
the second example we show a free theorem involving functions on lists. 
We show additional equivalences in the Appendix, including the
correctness of von Neumann's procedure for generating a fair sequence
of coin tosses from an unfair coin, and equivalences from the
recent papers using
bisimulations~\cite{Cruiblle:bisim-prob-CBV,DalLago:bisim-prob-CBN}.

We add dynamically allocated references to the language and extend the logical relation to the new
language in Section~\ref{sec:extension-references}. For simplicity we only sketch how
to extend the construction with first-order state. This already suggests that an extension
with general references can be done in the usual way for step-indexed logical
relations. We conclude the section by proving a representation independence result
involving both state and probabilistic choice.

All the references to the Appendix in this paper refer to appendix in the online
version~\cite{bizjak-birkedal:long-version}. 

\section{The language $\lang$}
\label{sec:language}

The language is a standard pure functional language with recursive, universal and
existential types with an additional choice primitive $\RAND$. The base types include the
type of natural numbers $\NAT$ with some primitive operations.
The grammar of terms $e$ is
\begin{align*}
  e  &\bnfeq x\hmid  \unitval\hmid \RAND e \hmid \underline{n}  \hmid \IFZERO{e}{e_1}{e_2} \hmid \PRED e
  \hmid \SUCC e \hmid \PAIR {e_1}{e_2}\hmid \PROJ{i}\,e
  \\ &\hspace{4ex}
  \hmid \FUN xe \hmid e_1\,e_2\hmid \INL\,e \hmid \INR\,e \hmid \MATCH{e}{x_1}{e_1}{x_2}{e_2}
  \hmid \TFUN e \hmid \TAPP e\\
  &\hspace{4ex} \hmid \PACK{e} \hmid
   \UNPACK{e_1}{x}{e_2} \hmid \FOLD{e} \hmid \UNFOLD{e}
\end{align*}%
We write $\underline n$ for the numeral
representing the natural
number $n$ and $\SUCC$ and $\PRED$ are the successor and predecessor functions,
respectively.
For convenience, numerals start at $\underline 1$. Given a numeral $\underline
n$, the term $\RAND \underline n$ evaluates to one of the numerals $\underline 1,\ldots,
\underline n$ with uniform probability. There are no types in the syntax of terms,
e.g., instead of $\Lambda \alpha . e$ and $e\,\tau$ we have $\TFUN{e}$ and $\TAPP{e}$. This
is for convenience only.

We write $\alpha, \beta, \ldots$ for \emph{type variables}
and $x, y,\ldots$ for \emph{term variables}. The notation
$\tau\subst{\vec\alpha}{\vec\tau}$ denotes the simultaneous capture-avoiding
substitution of types $\vec\tau$ for the free type variables $\vec\alpha$ in the type
$\tau$;
$e\subst{\vec x}{\vec v}$ denotes simultaneous capture-avoiding substitution of
values $\vec v$ for the free term variables $\vec x$ in the term $e$.

We write $\Stks$ for the set of evaluation contexts given by the call-by-value reduction
strategy. Given two evaluation contexts $E, E'$ we define their composition
$\ecomp{E}{E'}$ by induction on $E$ in the natural way. Given an evaluation context $E$
and expression $e$ we write $E[e]$ for the term obtained by plugging $e$ into $E$. For any
two evaluation contexts $E$ and $E'$ and a term $e$ we have $E[E'[e]] =
(\ecomp{E}{E'})[e]$.

For a type variable context $\Delta$, the judgment $\Delta \vdash \tau$ expresses that the
free type variables in $\tau$ are included in $\Delta$. 
The typing judgments are entirely standard with the addition of the
typing of $\RAND$ which is given by the rule
\begin{align*}
  \inferrule{\Delta\hmid \Gamma \vdash e : \NAT}{\Delta\hmid\Gamma\vdash \RAND\,e : \NAT}.
\end{align*}%
The complete set of typing rules are in the Appendix.
We write $\Type(\Delta)$ for the set of types well-formed in context $\Delta$, and $\Type$
for the set of \emph{closed} types $\tau$.  We write
$\Val{\tau}$ and $\Exp{\tau}$ for the sets of \emph{closed} values and terms of type
$\tau$, respectively.  We write $\Vals$ and $\Exps$ for the set of \emph{all}\footnote{In
  particular, we do not require them to be typeable.} \emph{closed}
values and closed terms, respectively. $\Stk{\tau}$ denotes the set of $\tau$-accepting
evaluation contexts, i.e., evaluation contexts $E$, such that given any closed term $e$ of
type $\tau$, $E[e]$ is a typeable term. $\Ectxt$ denotes the set of all evaluation
contexts.

For a typing context $\Gamma = x_1{:}\tau_1,\ldots,x_n{:}\tau_n$ with
$\tau_1,\ldots,\tau_n\in\Type$, let $\Subst(\Gamma)$
denote the set of type-respecting value substitutions, i.e. for all $i$, $\gamma(x_i) \in
\Val{\tau_i}$. 
In particular, if  $\Delta\hmid\Gamma\vdash e:\tau$ then $\emp\hmid\emp\vdash
e\gamma : \tau\delta$ for any $\delta\in \Type^\Delta$ and
$\gamma\in\Subst(\Gamma\delta)$, and the type system satisfies 
standard properties of progress and preservation and a canonical forms lemma.

The operational semantics of the language is a standard call-by-value semantics but
weighted with $p \in [0,1]$ which denotes the likelihood of that reduction. We write
$\stepsto{p}$ for the one-step reduction relation. All the usual
$\beta$ reductions have weight equal to $1$ and the reduction from $\RAND \underline n$ is
\vspace{-0.2cm}
\begin{align*}
  \RAND\underline{n} \stepsto{\frac{1}{n}} \underline k \qquad\ \
  \text{for } k \in \{1,2,\ldots,n\}.
\end{align*}
The rest of the rules are given in Fig.~\ref{app:fig:semantics} in the Appendix. The
operational semantics thus gives rise to a Markov chain with closed terms as states. In
particular for each term $e$ we have $\sum_{ e'\,\mid\,e \stepsto{p} e'}p \leq 1$.

\section{Observations and biorthogonality}
\label{sec:term-relat}

We will use biorthogonality to define the logical relation. This section provides the
necessary observation predicates used in the definition of the biorthogonal lifting of
value relations to expression relations. Because of the use of biorthogonality the value
relations (see Fig.~\ref{fig:logical-relation}) remain as simple as for a language
without probabilistic choice. The new quantitative aspects only appear in the definition
of the biorthogonal lifting ($\top\top$-closure) defined in
Section~\ref{sec:logical-relation}. Two kinds of observations are used.  The probability
of termination, $\PR{e}$, which is the actual probability that $e$ terminates, and its
approximation, the \emph{stratified} termination probability $\PRk{k}{e}$, where $k \in
\NN$ denotes, intuitively, the number of computation steps. The stratified termination
probability provides the link between steps in the operational semantics and the indexing in the
definition of the interpretation of types.

The probability of termination, $\PR{\cdot}$, is a function of type $\Exps \to \interval$
where $\interval$ is the unit interval $[0,1]$. Since $\interval$ is a pointed
$\omega$-cpo for the usual order, so is the space of all functions
$\Exps\to\interval$ with pointwise ordering. We define $\PR{\cdot}$ as a fixed point
of the continuous function $\Phi$ on this $\omega$-cpo:
Let $\Fl = \Exps \to \interval$ and define $\Phi : \Fl \to \Fl$ as
\begin{align*}
  \Phi(f)(e) =
  \begin{cases}
    1 &\text{if } e \in \Vals\\
    \displaystyle\sum_{e \stepsto{p}{e'}} p \cdot f\left(e'\right) & \text{ otherwise}
  \end{cases}
\end{align*}
Note that if $e$ is stuck then $\Phi(f)(e) = 0$ since the empty sum is $0$.

The function $\Phi$ is monotone and preserves suprema of $\omega$-chains. The proof is
straightforward and can be found in the Appendix. Thus $\Phi$ has a least
fixed point in $\Fl$ and we denote this fixed point by
$\PR{\cdot}$, i.e., $\PR{e} = \sup_{n\in\omega}\Phi^n(\bot)(e)$. 

To define the stratified observations we need the notion of a path.
Given terms $e$ and $e'$ a path $\pi$ from $e$ to $e'$, written $\redpath{\pi}{e}{e'}$, is
a sequence $e \stepsto{p_1} e_1 \stepsto{p_2} e_2 \stepsto{p_3} \cdots \stepsto{p_n} e'$.
The \emph{weight} $\weight{\pi}$ of a path $\pi$ is the product of the weights of
reductions in $\pi$.  We write $\paths$ for the set of all paths and $\cdot$ for their
concatenation (when defined).  For a non-empty path $\pi \in \paths$ we write $\last{\pi}$
for its last expression.

We call reductions of the form $\UNFOLD{(\FOLD{v})} \stepsto{1} v$ \emph{unfold-fold}
reductions and reductions of the form $\RAND\underline n \stepsto{\frac{1}{n}} \underline
k$ \emph{choice} reductions.  If \emph{none} of
the reductions in a path $\pi$ is a choice reduction we call $\pi$ \emph{choice-free} and
similarly if none of the reductions in $\pi$ is an unfold-fold reductions we call $\pi$
\emph{unfold-fold free}.

We define the following types of multi-step reductions which we use in the definition of
the logical relation.
\begin{itemize}
\item $e \stepstopure e'$ if there is a \emph{choice-free} path from $e$ to $e'$
\item $e \stepstozero e'$ if there is an \emph{unfold-fold} free path from $e$ to $e'$.
\item $e \stepstozeropure e'$ if $e \stepstopure e'$ and $e \stepstozero e'$.
\end{itemize}

The following useful lemma states that all but choice reductions preserve the
probability of termination. As a consequence, we will see that all but choice reductions
preserve equivalence. 
\begin{lemma}
  \label{lem:choice-free-reductions-dont-change}
  Let $e, e' \in \Exps$ and $e\stepstopure e'$. Then $\PR{e} = \PR{e'}$.
\end{lemma}
The proof proceeds on the length of the reduction path with the strengthened induction
hypothesis stating that the probabilities of termination of all elements on the path are
the same. To define the stratified probability of termination that approximates
$\PR{\cdot}$ we need an auxiliary notion.
\begin{definition}
  \label{def:red-tree}
  For a closed expression $e\in\Exps$ we define $\reds{e}$ as the (unique) set of paths
  containing \emph{exactly one} unfold-fold or choice reduction and \emph{ending} with such a
  reduction. More precisely, we define the function 
  $\redsname : \Exps \to \pow{\paths}$
  as the least function satisfying
  \begin{align*}
    \reds{e} = \begin{cases}
      \{ e \stepsto{1} e' \} &\text{if } e = E[\UNFOLD(\FOLD v)]\\
      \{ e \stepsto{p} E[\underline k] \isetsep p = \frac{1}{n}, k \in \{1,2,\ldots, n\} \} &
      \text{if } e = E[\RAND \underline{n}]\\
      \left\{ (e\stepsto{1} e')\cdot \pi \isetsep \pi \in \reds{e'}\right\} &
      \text{if } e \stepsto{1} e' \text{ and } e\stepstozeropure e'\\
      \emptyset &\text{otherwise}
    \end{cases}
  \end{align*}
  where we order the power set $\pow{\paths}$ by subset inclusion.
\end{definition}

Using $\reds{\cdot}$ we define a monotone map $\Psi : \Fl \to \Fl$
that preserves $\omega$-chains.
\begin{align*}
  \Psi(f)(e) =
  \begin{cases}
    1 &\text{if } \exists v \in \Vals, e \stepstozeropure v\\
    \displaystyle\sum_{\pi \in \reds{e}} \weight{\pi} \cdot f\left(\last{\pi}\right) &\text{otherwise}
  \end{cases}
\end{align*}
and then define $\PRk{k}{e} = \Psi^k(\bot)(e)$. The intended meaning of $\PRk{k}{e}$ is the
probability that $e$ terminates within $k$ unfold-fold and choice reductions. Since $\Psi$
is monotone we have that $\PRk{k}{e} \leq \PRk{k+1}{e}$ for any $k$ and $e$.

The following lemma is the reason for counting only certain reductions,
cf.\cite{Dreyer-Ahmed-Birkedal:LSLR}. It allows us to
stay at the same step-index even when taking steps in the operational semantics. As a
consequence we will get a more extensional logical relation. The proof is by case analysis
and can be found in the Appendix.
\begin{lemma}
  \label{lem:pure-reductions-dont-change}
  Let $e, e' \in \Exps$. If $e \stepstozeropure e'$ then for all $k$, $\PRk{k}{e} =
  \PRk{k}{e'}$.
\end{lemma}

The following is immediate from the definition of the chain
$\left\{\PRk{k}{e}\right\}_{k=0}^{\infty}$ and the fact that $\RAND\underline n$ reduces
with uniform probability.
\begin{lemma}
  \label{lem:prob-unfold-fold-choice}
  Let $e$ be a closed term. If $e \stepsto{1} e'$ and the reduction is an unfold-fold reduction then
    $\PRk{k+1}{e} = \PRk{k}{e'}$. If the reduction from $e$ is a choice reduction, then
    $\PRk{k+1}{e} = \frac{1}{|\reds{e}|}\sum_{\pi\in\reds{e}}\PRk{k}{\last{\pi}}$.
\end{lemma}

The following proposition is needed to prove adequacy of the logical relation with respect
to contextual equivalence. It is analogous to the property used to prove adequacy of
step-indexed logical relations for deterministic and nondeterministic languages. Consider
the case of may-equivalence. To prove adequacy in this case (cf. \cite[Theorem
4.8]{Birkedal-et-al:countable-nondet}) we use the fact
that if $e$ may-terminates, then there is a natural number $n$ such
that $e$ terminates in $n$ steps. This property does not hold 
in the probabilistic case, but the property analogous to it that is sufficient to prove
adequacy still holds.
\begin{proposition}
  \label{prop:prob-smaller}
  For each $e \in \Exps$ we have $\PR{e} \leq \sup_{k\in\omega}\left(\PRk{k}{e}\right)$.
\end{proposition}
\begin{proof}
  We only give a sketch; the full proof can be found in the Appendix.
  We use Scott induction on the set
  $\Sl = {\left\{ f \in \Fl \isetsep \forall e, f(e) \leq \sup_{k\in\omega}\left(\PRk{k}{e}\right) \right\}}$.
  It is easy to see that $\Sl$ is closed under limits of $\omega$-chains and that $\bot \in \Sl$
  so we only need to show that $\Sl$ is closed under $\Phi$. We can do this by considering
  the kinds of reductions from $e$ when considering $\Phi(f)(e)$ for $f \in \Sl$.
\end{proof}

\section{Logical, CIU and contextual approximation relations}
\label{sec:approx-relations}

The contextual and CIU (\textbf{c}losed \textbf{i}nstantiations of
\textbf{u}ses~\cite{Pitts:05}) approximations are defined in a way analogous to the one
for deterministic programming languages. We require some auxiliary notions.  A
\emph{type-indexed relation} $\mathrel{\R}$ is a set of tuples $(\Delta,\Gamma,e,e',\tau)$
such that $\Delta \vdash \Gamma$ and $\Delta \vdash \tau$ and
$\hastype{\Delta}{\Gamma}{e}{\tau}$ and $\hastype{\Delta}{\Gamma}{e'}{\tau}$.  We write
$\Delta\hmid\Gamma\vdash e\mathrel{\R} e':\tau$ for $(\Delta,\Gamma,e,e',\tau)\in\R$.

\begin{definition}
  [Precongruence]
  \label{def:tir-congruence}
A type-indexed relation \R is \emph{reflexive} if 
$\Delta\hmid\Gamma\vdash e:\tau$ implies $\Delta\hmid\Gamma\vdash e\mathrel\R e:\tau$. 
It is \emph{transitive} if 
$\Delta\hmid\Gamma\vdash e\mathrel\R e':\tau$ and $\Delta\hmid\Gamma\vdash e'\mathrel\R e'':\tau$ implies $\Delta\hmid\Gamma\vdash e\mathrel\R e'':\tau$. 
It is \emph{compatible} if it is closed under the term forming rules, e.g.,\footnote{We only show a
  few rules, the rest are analogous and can be found in the Appendix.}
\begin{align*}
  \inferrule{\Delta\hmid\Gamma,x{:}\tau_1\vdash e\mathrel\R e':\tau_2
  }{\Delta\hmid\Gamma\vdash\FUN xe\mathrel\R\FUN x{e'}:\tau_1\to\tau_2} \qquad
  \inferrule{\Delta\hmid\Gamma\vdash e \mathrel\R e': \NAT }{\Delta\hmid\Gamma\vdash \RAND
    e\mathrel\R\RAND e' :\NAT}
\end{align*}
A \emph{precongruence} is a reflexive, transitive and compatible type-indexed relation.
\end{definition}

The compatibility rules guarantee that a compatible relation is sufficiently big, i.e., at
least reflexive. In contrast, the notion of adequacy, which relates the operational
semantics with the relation, guarantees that it is not too big. In the deterministic case,
a relation $\Rl$ is adequate if when $e \mathrel\Rl e'$ are two related closed terms, then
if $e$ terminates so does $e'$. Here we need to compare probabilities of termination
instead, since these are our observations.
\begin{definition}
  \label{def:adequacy}
  A type-indexed relation $\R$ is \emph{adequate} if for all $e, e'$ such that
  $\emp\hmid\emp \vdash e \mathrel\R e' : \tau$ we have
  $\PR{e} \leq \PR{e'}$.
\end{definition}
The \emph{contextual approximation relation}, written $\ctxapprox{\Delta}{\Gamma}{e}{e'}{\tau}$,
is defined as the \emph{largest adequate precongruence} and the \emph{CIU approximation
  relation}, written $\ciuapprox{\Delta}{\Gamma}{e}{e'}{\tau}$, is 
defined using evaluation contexts in the usual way, e.g.~\cite{Pitts:05}, using $\PR{\cdot}$ for observations.
The fact that the largest adequate precongruence exists is proved as in~\cite{Pitts:05}.

\subsubsection{Logical relation}
\label{sec:logical-relation}

We now define the step-indexed logical relation. We present the construction in the
elementary way with explicit indexing instead of using a logic with guarded recursion as
in~\cite{Dreyer-Ahmed-Birkedal:LSLR} to remain self-contained.

Interpretations of types will be defined as decreasing sequences of relations on
\emph{typeable} values. For \emph{closed types} $\tau$ and $\sigma$ we define the sets
$\VRel{\tau}{\sigma}$, $\SRel{\tau}{\sigma}$ and $\ERel{\tau}{\sigma}$ to be the sets of
decreasing sequences of relations on typeable values, evaluation contexts and expressions
respectively. The types $\tau$ and $\sigma$ denote the types of the left-hand side and the
right-hand side respectively, i.e. if $(v, u) \in \phi(n)$ for $\phi \in
\VRel{\tau}{\sigma}$ then $v$ has type $\tau$ and $u$ has type $\sigma$.
The order relation $\leq$ on these sets is defined pointwise, e.g. for
$\phi, \psi \in \VRel{\tau}{\sigma}$ we write $\phi \leq \psi$ if $\forall
n \in \NN, \phi(n) \subseteq \psi(n)$. We implicitly use the inclusion from
$\VRel{\tau}{\sigma}$ to $\ERel{\tau}{\sigma}$. The reason for having relations on values
and terms of different types on the left and right-hand sides is so we are able to prove
parametricity properties in Section~\ref{sec:examples}.

We define maps $\T{\cdot}_{\tau,\sigma} : \VRel{\tau}{\sigma} \to \SRel{\tau}{\sigma}$ and
$\Te{\cdot}_{\tau,\sigma} : \SRel{\tau}{\sigma} \to \ERel{\tau}{\sigma}$. We usually omit
the type indices when they can be inferred from the context. The maps are defined
as follows
\begin{align*}
  \T{r}_{\tau,\sigma}(n) &= \left\{(E, E') \isetsep
    \forall k \leq n, \forall (v, v') \in r(k), \PRk{k}{E[v]} \leq \PR{E'[v']}\right\}
\end{align*}
and $\Te{r}_{\tau,\sigma}(n) = \left\{(e, e') \isetsep
  \forall k \leq n, \forall (E, E') \in r(k), \PRk{k}{E[e]} \leq \PR{E'[e']}\right\}.$
Note that we only count steps evaluating the left term in defining $\T{r}$ and $\Te{r}$.
We write $\TT{r} = \Te{\T{r}}$ for their composition from $\VRel{\tau}{\sigma}$ to
$\ERel{\tau}{\sigma}$.
The function $\T{\cdot}$ is order-reversing and $\TT{\cdot}$ is
order-preserving and inflationary.
\begin{lemma}
  \label{lem:tt-closure}
  Let $\tau, \sigma$ be closed types and $r, s \in \VRel{\tau}{\sigma}$. Then $r \leq
  \TT{r}$ and if $r \leq s$ then $\T{s} \leq \T{r}$ and $\TT{r} \leq \TT{s}$.
\end{lemma}

For a type-variable context $\Delta$ we define $\VRelD{\Delta}$ using
$\VRel{\cdot}{\cdot}$ as
\begin{align*}
  \VRelD{\Delta} = \left\{(\phi_1, \phi_2, \phi_r) \isetsep \phi_1, \phi_2 \in \Type^\Delta,
  \forall \alpha \in \Delta, \phi_r(\alpha) \in \VRel{\phi_1(\alpha)}{\phi_2(\alpha)}
                          \right\}
\end{align*}
where the first two components give syntactic types for the left and right hand sides of
the relation and the third component is a relation between those types.

The interpretation of types, $\den{\cdot \vdash \cdot}$ is by induction on the judgement
$\Delta \vdash \tau$. For a judgment $\Delta \vdash \tau$ and $\phi \in \VRelD{\Delta}$
we have $\den{\Delta \vdash \tau}(\phi) \in \VRel{\phi_1(\tau)}{\phi_2(\tau)}$ where the
$\phi_1$ and $\phi_2$ are the first two components of $\phi$ and $\phi_1(\tau)$ denotes
substitution. Moreover
$\den{\cdot}$ is \emph{non-expansive} in the sense that $\den{\Delta\vdash\tau}(\phi)(n)$
can depend only on the values of $\phi_r(\alpha)(k)$ for $k \leq
n$, see \cite{Birkedal:Reus:Schwinghammer:Stovring:Thamsborg:Yang:11} for this metric view
of step-indexing. The interpretation of
types is defined in Fig.~\ref{fig:logical-relation}. Observe that the value relations are
as simple as for a language without probabilistic choice. The crucial difference is hidden
in the $\top\top$-closure of value relations.

\begin{figure}[htb]
  \vspace{-0.6cm}
  \centering
  \begin{align*}
    \den{\Delta \vdash \NAT}(\phi)(n) &= \left\{(\underline k, \underline k) \isetsep k
                                        \in \NN, k > 0\right\} \\
    \den{\Delta \vdash \tau \to \sigma}(\phi)(n) &= 
      \begin{array}[t]{l}
        \{ \left(\FUN{x}{e},
          \FUN{y}{e'}\right) \isetsep \forall j \leq n, 
          \forall (v, v') \in \den{\Delta \vdash \tau}(\phi)(j), \\
          \quad\qquad\qquad\qquad ((\FUN{x}{e})\,v, (\FUN{y}{e'})\,v') \in
        \TT{\den{\Delta\vdash \sigma}(\phi)}(j)
        \}
     \end{array} 
    \\
    \den{\Delta \vdash \ALL{\alpha}{\tau}}(\phi)(n) &=
    \begin{array}[t]{l}
    \{ \left(\TFUN{e},
    \TFUN{e'}\right) \isetsep \forall \sigma, \sigma' \in \Type,
     \forall r \in \VRel{\sigma}{\sigma'}, \\
     \qquad
     (e, e') \in \TT{\den{\Delta,\alpha\vdash
         \tau}\left(\extend{\phi}{\alpha}{r}\right)}(n)
     \}
     \end{array}
     \\
     \den{\Delta \vdash \EX{\alpha}{\tau}}(\phi)(n) &=
     \begin{array}[t]{l}
       \{\left(\PACK{v},
    \PACK{v'}\right) \isetsep \exists \sigma, \sigma' \in \Type,
     \exists r \in \VRel{\sigma}{\sigma'}, \\
     \quad\qquad\qquad\qquad\qquad
     (v, v') \in \den{\Delta,\alpha\vdash
       \tau}\left(\extend{\phi}{\alpha}{r}\right)(n) \}
   \end{array}\\
    \den{\Delta \vdash \REC{\alpha}{\tau}}(\phi)(0) &=
         \Val{\phi_1(\REC{\alpha}{\tau})}\times \Val{\phi_2(\REC{\alpha}{\tau})}\\
    \den{\Delta \vdash \REC{\alpha}{\tau}}(\phi)(n+1) &= 
    \begin{array}[t]{l}
      \{\left(\FOLD{v},
    \FOLD{v'}\right) \isetsep \\
    \qquad\qquad  (v, v') \in \den{\Delta,\alpha\vdash
         \tau}\left(\extend{\phi}{\alpha}{\den{\Delta\vdash
             \REC{\alpha}{\tau}}(\phi)}\right)(n) \}
    \end{array}
 \end{align*}
 \vspace{-0.7cm}
  \caption{Interpretation of types. The cases for sum and product types are in Appendix.}
  \label{fig:logical-relation}
  \vspace{-0.82cm}
\end{figure}

\paragraph{Context extension lemmas}
To prove soundness and completeness we need lemmas stating how extending evaluation
contexts preserves relatedness. We only show the case for $\RAND$. The rest are similarly
simple.
\begin{lemma}
  \label{lem:ctx-extend-rand}
  Let $n\in\NN$.
  If $(E, E') \in \T{\den{\Delta \vdash \NAT}(\phi)}(n)$ are related evaluation contexts then 
  $\left(\ecomp{E}{(\RAND{[]})}, \ecomp{E'}{(\RAND{[]})}\right) \in 
  \T{\den{\Delta \vdash \NAT}(\phi)}(n)$.
\end{lemma}
\begin{proof}
  Let $n\in\NN$ and $(v, v') \in\den{\Delta\vdash\tau}(\phi)(n)$. By construction we have
  $v = v' = \underline m$ for some $m \in \NN$, $m \geq 1$. Let $k \leq n$. If $k=0$
  the result is immediate, so assume $k = \ell + 1$. Using
  Lemma~\ref{lem:prob-unfold-fold-choice} we have
  $\PRk{k}{E[\RAND{\underline{m}}]} = \frac{1}{m}\sum_{i=1}^{m}
  \PRk{\ell}{E[\underline{i}]}$
  and using the assumption $(E, E') \in \T{\den{\Delta \vdash \NAT}(\phi)}(n)$,
  the fact that $k\leq n$ and monotonicity in the step-index the latter term is 
  less than $\frac{1}{m}\sum_{i=1}^m\PR{E'[\underline{i}]}$ which by definition of $\PR{\cdot}$
  is equal to $\PR{E'[\RAND{\underline{m}}]}$.
\end{proof}

We define the logical approximation relation for open terms given the
interpretations of types in Fig.~\ref{fig:logical-relation}. We define
$\logapprox{\Delta}{\Gamma}{e}{e'}{\tau}$ to mean
\begin{align*}
  \forall n \in \NN, \forall \phi \in \VRelD{\Delta},\forall
  (\gamma,\gamma')\in\den{\Delta\vdash\Gamma}(\phi)(n), (e\gamma, e'\gamma) \in\TT{\den{\Delta\vdash \tau}{\phi}}(n)
\end{align*}
Here $\den{\Delta\vdash\Gamma}$ is the obvious extension of interpretation of types to
interpretation of contexts which relates substitutions, mapping variables to values.
We have
\begin{proposition}[Fundamental property]
  \label{prop:logrel-is-compatible}
  The logical approximation relation $\logapproxrel$ is compatible. In particular it is reflexive.
\end{proposition}
\begin{proof}
  The proof is a simple consequence of the context extension lemmas. We show the case for
  $\RAND$. We have to show that $\logapprox{\Delta}{\Gamma}{e}{e'}{\NAT}$ implies
  $\logapprox{\Delta}{\Gamma}{\RAND\,e}{\RAND\,e'}{\NAT}$.
  Let $n\in\NN$, $\phi \in \VRelD{\Delta}$ and $(\gamma, \gamma') \in \den{\Delta \vdash
    \Gamma}(\phi)(n)$.
  Let $f = e\gamma$ and $f' = e'\gamma'$. Then our assumption gives us $(f, f') \in
  \TT{\den{\Delta \vdash \NAT}(\phi)}(n)$ and we are to show
  $(\RAND\,f, \RAND\,f') \in \TT{\den{\Delta \vdash \NAT}(\phi)}(n)$. Let $j \leq n$ and
  $(E, E') \in \T{\den{\Delta\vdash\NAT}(\phi)}(j)$. Then from
  Lemma~\ref{lem:ctx-extend-rand} we have
  $\left(\ecomp{E}{(\RAND{[]})}, \ecomp{E'}{(\RAND{[]})}\right) \in 
  \T{\den{\Delta \vdash \NAT}(\phi)}(j)$ which suffices by the definition of 
  the orthogonality relation and the assumption
  $(f, f') \in \TT{\den{\Delta \vdash \NAT}(\phi)}(n)$.
\end{proof}

We now want to relate logical, CIU and contextual approximation relations. 
\begin{corollary}
  \label{cor:logrel-is-adequat}
  Logical approximation relation $\logapproxrel$ is adequate.
\end{corollary}
\begin{proof}
  Assume $\logapprox{\emp}{\emp}{e}{e'}{\tau}$. We are to show that $\PR{e} \leq
  \PR{e'}$. Straight from the definition we have $\forall n \in \NN, (e, e') \in \TT{\den{\emp \vdash \tau}}(n).$
  The empty evaluation context is always related to itself (at any
  type). This implies $\forall n \in \NN, \PRk{n}{e} \leq \PR{e'}$
  which further implies (since the right-hand side is independent of $n$) that 
  $\sup_{n\in\omega}\left(\PRk{n}{e}\right) \leq \PR{e'}$.
  Using Proposition~\ref{prop:prob-smaller} we thus have
  $\PR{e} \leq \sup_{n\in\omega}\left(\PRk{n}{e}\right) \leq \PR{e'}$
  concluding the proof.
\end{proof}

We now have that the logical relation is adequate and compatible. This does not
immediately imply that it is contained in the contextual approximation relation, since we
do not know that it is transitive. However we have the following lemma where by transitive
closure we mean that for each $\Delta$, $\Gamma$ and $\tau$ we take the transitive closure
of the relation $\{ (e, e') \isetsep \logapprox{\Delta}{\Gamma}{e}{e'}{\tau}\}$. This is
another type-indexed relation.
\begin{lemma}
  \label{lem:trans-closure-compat}
  The transitive closure of $\logapproxrel$ is compatible and adequate.
\end{lemma}
\begin{proof}
  Transitive closure of an adequate relation is adequate. Similarly the transitive closure
  of a compatible and \emph{reflexive} relation (in the sense of
  Definition~\ref{def:tir-congruence}) is again compatible (and reflexive).
\end{proof}
\begin{theorem}[CIU theorem]
  \label{thm:CIU-theorem}
  The relations $\logapproxrel$, $\ciuapproxrel$ and $\ctxapproxrel$ coincide.
\end{theorem}
\begin{proof}
  It is standard (e.g.~\cite{Pitts:05}) that $\ctxapproxrel$ is included
  in $\ciuapproxrel$. We show that the logical approximation relation is
  contained in the CIU approximation relation in the standard way for biorthogonal
  step-indexed logical relations. To see that 
  $\logapproxrel$ is included in $\ctxapproxrel$ we have by
  Lemma~\ref{lem:trans-closure-compat} that the transitive closure of $\logapproxrel$ is
  an adequate precongruence, thus included in $\ctxapproxrel$. And $\logapproxrel$ is
  included in the transitive closure of $\logapproxrel$.
  Corollary~\ref{app:cor:logrel-ciu-extend} in the appendix completes the cycle of inclusions.
\end{proof}

Using the logical relation and Theorem~\ref{thm:CIU-theorem} we can
prove some extensionality properties.
The proofs are standard and can be found in the Appendix.
\begin{lemma}[Functional extensionality for values]
  \label{lem:func-extensionality}
  Suppose $\tau, \sigma \in \Type(\Delta)$ and let $f$ and $f'$ be two
  \emph{values} of type $\tau\to\sigma$ in context $\Delta\ |\ \Gamma$.
  If for all $u \in \Val{\tau}$ we have
  $\ctxapprox{\Delta}{\Gamma}{f\,u}{f'\,u}{\sigma}$
  then
  ${\ctxapprox{\Delta}{\Gamma}{f}{f'}{\tau\to\sigma}}$.
\end{lemma}
The extensionality for \emph{expressions}, as opposed to only \emph{values}, of function
type does not hold in general due to the presence of choice reductions. See
Remark~\ref{rem:failure-extensionality} for an example.
We also have extensionality for \emph{values} of universal types.
\begin{lemma}[Extensionality for the universal type]
  \label{lem:all-extensionality}
  Let $\tau \in \Type(\Delta, \alpha)$ be a type.
  Let $f, f'$ be two \emph{values} of type $\ALL{\alpha}{\tau}$ in context $\Delta\ |\
  \Gamma$. If for all closed types $\sigma$ we have
  ${\ctxapprox{\Delta}{\Gamma}{\TAPP{f}}{\TAPP{f'}}{\tau\subst{\alpha}{\sigma}}}$
  then 
  ${\ctxapprox{\Delta}{\Gamma}{f}{f'}{\ALL{\alpha}{\tau}}}$.
\end{lemma}

\section{Examples}
\label{sec:examples}

We now use our logical relation to prove 
some example equivalences.  We show two examples involving
polymorphism. In the Appendix we show additional examples. In particular
we show the correctness of von Neumann's procedure for generating a fair sequence of coin
tosses from an unfair coin. That example in particular shows how the use of
biorthogonality allows us to ``externalize'' the reasoning to arithmetic
manipulations.

We first define 
$\FIX :
\ALL{\alpha,\beta}{((\alpha\mathop\to\beta)\mathop\to(\alpha\mathop\to\beta))
  \to(\alpha\mathop\to\beta)}$
be the term
$\TFUN{\TFUN{\FUN{f}{\FUN{z}{\delta_f(\FOLD{\delta_f})\,z}}}}$ where $\delta_f$ is the
term $\FUN{y}{\syn{let}~y'~=\UNFOLD{y}~\syn{in}~f\,(\FUN{x}{y'\,y\,x})}$.
This is a call-by-value fixed-point combinator. We also write
$e_1 \oplus e_2$ for the term $\IFZERO{\RAND\underline 2}{e_1}{e_2}$. Note that the choice
is made before evaluating $e_i$'s.

We characterize inhabitants of a polymorphic
type and show a free theorem. For the former, we need to know which real numbers can be 
probabilities of termination of programs.
Recall that a real number $r$ is \emph{left-computable} if there exists a \emph{computable}
increasing (not necessarily strictly) sequence $\{q_n\}_{n\in\omega}$ of \emph{rational numbers}
such that $r = \sup_{n\in\omega} q_n$. In Appendix~\ref{app:sec:prob-conv} we prove
\begin{proposition}
  \label{prop:characterizing-prob-conv}
  For any expression $e$, $\PR{e}$ is a left-computable real number and for any
  left-computable real number $r$ in the interval $[0,1]$ there is a closed term $e_r$ of
  type $\ONE \to \ONE$ such that $\PR{e_r\,\unitval} = r$.
\end{proposition}

\subsubsection{Inhabitants of the type $\ALL{\alpha}{\alpha\to\alpha}$}
\label{sec:inhabitants-forall}

In this section we use further syntactic sugar for sequencing. When $e, e' \in \Exps$ are
closed terms we write $e ; e'$ for $\left(\FUN{\_}{e'}\right)e$, i.e. first run $e$,
ignore the result and then run $e'$. We will need the property that
for all terms $e, e' \in \Exps$, $\PR{e; e'} = \PR{e} \cdot \PR{e'}$.
The proof is by Scott induction and can be found in the Appendix.

Using Proposition~\ref{prop:characterizing-prob-conv} we have for each left-computable real $r$ in
the interval $[0,1]$ an inhabitant $t_r$ of the type $\ALL{\alpha}{\alpha\to\alpha}$ given by
$\TFUN{\FUN{x}{e_r\,\unitval; x}}$.

We now show that these are the only inhabitants of $\ALL{\alpha}{\alpha\to\alpha}$ of the
form $\TFUN{\FUN{x}{e}}$. Given such an inhabitant let $r = \PR{e\subst{x}{\unitval}}$. We know
from Proposition~\ref{prop:characterizing-prob-conv} that $r$ is left-computable.

Given a value $v$ of type $\tau$
and $n\in\NN$ we define relations $R(n) = \{(\unitval, v)\}$ and $S(n) = \{(v,
\unitval)\}$. Note that the relations are independent of $n$, i.e. $R$ and $S$ are
constant relations. By reflexivity of the logical relation and the relational actions of types
we have
\begin{align}
  \label{eq:forall}
  \forall n, (e\subst{x}{\unitval}, e\subst{x}{v}) \in \TT{R}(n)\qquad\text{and}\qquad
  \forall n, (e\subst{x}{v}, e\subst{x}{\unitval}) \in \TT{S}(n)
\end{align}
from which we conclude that $\PR{e\subst{x}{\unitval}} = \PR{e\subst{x}{v}}$.  We
now show that $v$ and $e\subst{x}{v}$ are CIU-equivalent.
Let $E \in \Stk{\tau}$ be an evaluation context.
Let $q = \PR{E[v]}$. Define the evaluation context $E' = -; e_q\,\unitval$. Then
$(E, E') \in \T{S}(n)$ for all $n$ which then means, using \eqref{eq:forall}
and Proposition~\ref{prop:prob-smaller}, that $\PR{E[e\subst{x}{v}]} \leq
\PR{E'[e\subst{x}{\unitval}]}$. We then have
\[\PR{E'[e\subst{x}{\unitval}]} = \PR{e\subst{x}{\unitval}}\cdot\PR{e_q\,\unitval} = r
\cdot \PR{E[v]}\] and so $\PR{E[e\subst{x}{v}]} \leq r \cdot \PR{E[v]}$.

Similarly we have $(E', E) \in \T{R}(n)$ for all $n$ which implies
$\PR{E[e\subst{x}{v}]} \geq \PR{E'[e\subst{x}{\unitval}]}$. We also have
$\PR{E'[e\subst{x}{\unitval}]} = r \cdot \PR{E[v]}$.

So we have proved $\PR{E[e\subst{x}{v}]} = r\cdot\PR{E[v]} =
\PR{e\subst{x}{v}}\cdot\PR{E[v]}$. It is easy to show by Scott induction,
that $\PR{E[\TAPP{t_r}\,v]} = \PR{e_r\,\unitval}\cdot\PR{E[v]}$. We have thus shown that for
any value $v$, the terms $e\subst{x}{v}$ and $\PR{\TAPP{t_r}\,v}$ are CIU-equivalent. Using
Theorem~\ref{thm:CIU-theorem} and
Lemmas~\ref{lem:all-extensionality}~and~\ref{lem:func-extensionality}
we conclude that the terms $\ALL{\alpha}{\FUN{x}{e}}$
and $t_r$ are contextually equivalent.

\begin{remark}
  \label{rem:failure-extensionality}
  Unfortunately we cannot so easily characterize general values of the type
  $\ALL{\alpha}{\alpha\to\alpha}$, that is, those not of the form $\TFUN{v}$ for a value
  $v$. Consider the term $\TFUN{t_{\frac{1}{2}}\oplus t_{\frac{1}{3}}}$. It is a
  straightforward calculation that for any evaluation context $E$ and value $v$, 
  $\PR{E\left[\left(t_{\frac{1}{2}}\oplus t_{\frac{1}{3}}\right)\,v\right]} =
    \frac{5}{12}\PR{E[v]} = \PR{E\left[t_{\frac{5}{12}}\,v\right]}$
  thus if $\TFUN{t_{\frac{1}{2}}\oplus t_{\frac{1}{3}}}$ is equivalent to any $\TFUN{t_r}$
  it must be $\TFUN{t_{\frac{5}{12}}}$.

  Let $E$ be the evaluation context
  $E = \syn{let}~f~=\TAPP{\empeval}~\syn{in}~
                      \syn{let}~x~=f\,\unitval~\syn{in}~f\,\unitval$.
  We compute $\PR{E\left[\TFUN{t_{\frac{1}{2}}\oplus t_{\frac{1}{3}}}\right]} = \frac{13}{72}$
  and $\PR{E\left[\TFUN{t_{\frac{5}{12}}}\right]} = \frac{25}{144}$
  showing that $\TFUN{t_{\frac{1}{2}}\oplus t_{\frac{1}{3}}}$ is \emph{not} equivalent to
  $\TFUN{t_{\frac{5}{12}}}$.

  This example also shows that extensionality for \emph{expressions}, as opposed to
  \emph{values}, of function type does not hold. The reason is that probabilistic choice
  is a computational effect and so it matters how many times we evaluate the term and
  this is what the constructed evaluation context uses to distinguish the terms.
\end{remark}

\subsubsection{A free theorem for lists}
\label{sec:free-theorem}

Let $\tau$ be a type and $\alpha$ not free
in $\tau$. We write $[\tau]$ for the type of lists $\REC{\alpha}{(\ONE + \tau \times
  \alpha)}$, $\syn{nil}$ for the empty list and $\syn{cons} : \ALL{\alpha}{\alpha \to
  [\alpha] \to [\alpha]}$ for the other constructor
${\syn{cons} = \TFUN{\FUN{x}{\FUN{xs}{\FOLD{(\INR{\PAIR{x}{xs}})}}}}}$.
The function $\syn{map}$ of type $\ALL{\alpha}{\ALL{\beta}{(\alpha \to \beta) \to [\alpha]
    \to [\beta]}}$ is the function applying the given function to all elements of
the list in order. Additionally, we define composition of terms $f \comp g$ as the term
$\FUN{x}{f(g(x))}$ (for $x$ not free in $f$ and $g$).

We will now show that any term $m$ of type $\ALL{\alpha}{\ALL{\beta}{(\alpha \to \beta)
    \to [\alpha] \to [\beta]}}$ equivalent to a term of the form
$\TFUN{\TFUN{\FUN{x}{e}}}$ satisfies
${\TAPP{\TAPP{m}}\,(f \comp g) \ctxequivrel \TAPP{\TAPP{m}}{f} \comp \TAPP{\TAPP{\syn{map}}}\,g}$
for all \emph{values} $f$ and all \emph{deterministic and terminating} $g$. By this we
mean that for each value $v$ in the domain of $g$, there exists a \emph{value} $u$ in the
codomain of $g$, such that $g\,v \ctxequivrel u$. For instance, if $g$ reduces without
using choice reductions and is terminating, then $g$ is deterministic. There are other
functions that are also deterministic and terminating, though, for instance
$\FUN{x}{\unitval \oplus \unitval}$.
In the Appendix we show that these restrictions are not superfluous.

So let $m$ be a closed term of type
${\ALL{\alpha}{\ALL{\beta}{(\alpha \to \beta) \to [\alpha] \to [\beta]}}}$
and suppose further that $m$ is equivalent to a term of the form $\TFUN{\TFUN{\FUN{x}{e}}}$.
Let $\tau, \sigma, \rho \in \Type$ be closed types and $f \in \Val{\sigma\to\rho}$ and $g
\in \Exp{\tau\to\sigma}$ be a deterministic and terminating function. Then
\begin{align*}
  \ctxequiv{\emp}{\emp}{\TAPP{\TAPP{m}}(f \comp g)}{\TAPP{\TAPP{m}} f \comp
  \TAPP{\TAPP{\syn{map}}} g}{[\tau] \to [\rho]}.
\end{align*}
We prove two approximations separately, starting with $\ctxapproxrel$. We use
Theorem~\ref{thm:CIU-theorem} multiple times. We have
$\hastype{\alpha, \beta}{\emp}{\TAPP{\TAPP{m}}}{(\alpha \to \beta)\to [\alpha] \to [\beta]}.$
Let $R = \lambda n . \{ (v, u) \isetsep g\,v \ctxequivrel u \}$ be a member of
$\VRel{\tau}{\sigma}$ and $S \in \VRel{\rho}{\rho}$ be the constant identity relation on
$\Val{\rho}$.  Let $\phi$ map $\alpha$ to $R$ and $\beta$ to $S$.
Proposition~\ref{prop:logrel-is-compatible} gives
${(\TAPP{\TAPP{m}},
\TAPP{\TAPP{m}}) \in \TT{\den{(\alpha \to \beta) \to [\alpha]\to [\beta]}(\phi)}(n)}$
for all $n\in\NN$.

We first claim that $(f \comp g, f) \in \den{\alpha \to \beta}(\phi)(n)$ for all $n \in \NN$.
Since $f$ is a value and has a type, it must be of the form $\FUN{x}{e}$ for some $x$ and $e$.
Take $j \in \NN$, related values $(v, u) \in r(j)$, $k \leq j$ and $(E, E') \in \T{S}(k)$
two related evaluation contexts. We then have $\PR{E'[f\,u]} = \PR{E'[f(g\,v)]}$
by Theorem~\ref{thm:CIU-theorem} and the definition of relation $R$.
Using the results about $\PRk{k}{\cdot}$ and $\PR{\cdot}$ proved in
Section~\ref{app:sec:distributions} in the Appendix this gives us
\begin{align*}
  \PRk{k}{E[f(g(v))]} &\leq \sum_{\makebox[30pt]{$\scriptstyle\redpath{\pi}{f(g(v))}{w}$}}\weight{\pi}\PRk{k}{E[w]} \leq  \sum_{\makebox[30pt]{$\scriptstyle\redpath{\pi}{f(g(v))}{w}$}}\weight{\pi}\PR{E'[w]}
\end{align*}
and the last term is equal to $\PR{E'[f(g\,v)]}$ which is equal to $\PR{E'[f\,u]}$.

From this we can conclude
$(\TAPP{\TAPP{m}}\,(f \comp g),  \TAPP{\TAPP{m}}\,f) \in \TT{\den{[\alpha]\to [\beta]}(\phi)}(n)$
for all $n\in\NN$. Note that we have \emph{not yet} used the fact that $g$ is deterministic and
terminating. We do so now.

Let $xs$ be a list of elements of type $\tau$. Then induction on
the length of $xs$, using the assumption on $g$, we can derive that there exists a list $ys$ of
elements of type $\sigma$, such that $\TAPP{\TAPP{\syn{map}}}\,g\,xs \ctxequivrel ys$ and
$(xs, ys) \in \den{[\alpha]}(\phi)(n)$ for all $n$.

This gives us
${(\TAPP{\TAPP{m}}\,(f \comp g)\,xs,  \TAPP{\TAPP{m}}\,f\,ys) \in \TT{\den{[\beta]}(\phi)}(n)}$
for all $n\in\NN$. Since the relation $S$ is the identity relation we have
for all evaluation contexts $E$ of a suitable type, $(E, E) \in \T{S}(n)$ for all $n$,
which gives
\begin{align*}
  \TAPP{\TAPP{m}}\,(f \comp g)\,xs \ciuapproxrel  \TAPP{\TAPP{m}}\,f\,ys &\ctxequivrel
  \TAPP{\TAPP{m}}\,f\,(\TAPP{\TAPP{\syn{map}}}\,g\,xs)
  \ctxequivrel (\TAPP{\TAPP{m}}\,f\comp \TAPP{\TAPP{\syn{map}}}\,g)\, xs
\end{align*}
where the last equality holds because $\beta$-reduction is an equivalence.

We now conclude by using the fact that $m$ is (equivalent to) a term of the form
$\TFUN{\TFUN{\FUN{x}{e}}}$ and use Lemma~\ref{lem:func-extensionality} to conclude
${\TAPP{\TAPP{m}}\,(f \comp g) \ctxapproxrel \TAPP{\TAPP{m}}\,f\comp \TAPP{\TAPP{\syn{map}}}\,g}.$

For the other direction, we proceed analogously. The relation for $\beta$ remains the
identity relation, and the relation for $R$ for $\alpha$ is $\{(v, u) \isetsep v \ctxequivrel g\,u\}$.

\section{Extension to references}
\label{sec:extension-references}

We now sketch the extension of $\lang$ to include dynamically allocated references. For
simplicity we add ground store only, so we do not have to solve a domain equation
giving us the space of semantic types and worlds~\cite{Ahmed:phdthesis}.
We show an equivalence using
state and probabilistic choice which shows that the addition
of references to the language is orthogonal to the addition of probabilistic choice.
We conjecture that the extension with \emph{higher-order} dynamically allocated references
can be done as in earlier work on step-indexed logical
relations~\cite{Birkedal-et-al:impact-hos-local}.

We extend the language by adding the type $\REFTY\,\NAT$ and extend the grammar of terms with
$\ell \hmid \REF{e} \hmid \ASSIGN{e_1}{e_2} \hmid \LOOKUP{e}$ with $\ell$ being locations.

To model allocation we need to index the interpretation of types by worlds. To keep things
simple
a world $w\in \Wl$ is partial bijection $f$ on locations
together with, for each pair of locations $(\ell_1, \ell_2) \in f$, a relation
$R$ on numerals. We write $(\ell_1, \ell_2, R) \in w$ when the partial bijection in $w$
relates $\ell_1$ and $\ell_2$ and $R$ is the relation assigned to the pair $(\ell_1,
\ell_2)$. Technically, worlds are relations of type $\LOC^2 \times \pow{\{ \underline n
\,\mid\, n \in \NN\}}$ satisfying the conditions described above.

The operational semantics has to be extended to include heaps, which are modeled as finite
maps from locations to numerals. A pair of heaps $(h_1,h_2)$ satisfies the world $w$,
written $(h_1,h_2) \in\sat{w}$, when
$\forall (\ell_1, \ell_2, R) \in w, \left(h_1(\ell_1), h_2(\ell_2)\right) \in R$.
The interpretation of types is then extended to include worlds. The denotation of a type
is now an element of $\Wl \overset{mon}{\to} \VRel{\cdot}{\cdot}$ where the order on $\Wl$
is inclusion. Let $\WRel{\tau}{\tau'} = \Wl \overset{mon}{\to} \VRel{\tau}{\tau'}$. We
define ${\den{\Delta \vdash \REF\,\NAT}(\phi)(n)}$ as
${\lambda w . \left\{ (\ell_1, \ell_2) \isetsep (\ell_1, \ell_2, =) \in w \right\}}$
where $=$ is the equality relation on numerals.

The rest of the interpretation stays the same, apart from some quantification over
``future worlds'' in the function case to maintain monotonicity. We also need to change
the definition of the $\top\top$-closure to use the world satisfaction relation. For $r
\in \WRel{\tau}{\tau'}$ we define an indexed relation (indexed by worlds) $\T{r}$ as
\[
  \T{r}(w)(n)
  \left\{
         (E, E') \ \middle|\ \begin{array}{c}
         \forall w' \geq w,
         \forall k \leq n,
          \forall (h_1, h_2) \in \sat{w'}, \forall v_1, v_2 \in r(w')(k),\\
         \PRk{k}{\conf{h_1}{E[v_1]}} \leq \PR{\conf{h_2}{E[v_2]}}
          \end{array} \right\}
\]
and analogously for $\Te{\cdot}$.

We now sketch a proof that two modules, each implementing a counter by using a single
internal location, are contextually equivalent. The increment method is special. When
called, it chooses, uniformly, whether to increment the counter or not. The two modules
differ in the way they increment the counter. One module increments the counter by $1$,
the other by $2$.
Concretely, we show that the two counters
${\PACK{\left(\FUN{-}{\REF{\underline 1}}, \FUN{x}{\LOOKUP{x}}, \FUN{x}{\unitval \oplus \left(\ASSIGN{x}{\SUCC{\LOOKUP{x}}}\right)}\right)}}$
and
${\PACK{\left(\FUN{-}{\REF{\underline 2}}, \FUN{x}{\LOOKUP{x}~\syn{div}~\underline{2}},
  \FUN{x}{\unitval \oplus \left(\ASSIGN{x}{\SUCC({\SUCC{\LOOKUP{x}}})}\right)}\right)}}$
are contextually equivalent at type
$\EX{\alpha}{(\ONE \to \alpha) \times (\alpha \to \NAT) \times (\alpha \to \ONE)}.$
We have used $\syn{div}$ for the division function on numerals which can easily be
implemented.

The interpretation of existentials $\den{\Delta \vdash \EX{\alpha}{\tau}}(\phi)(n)$ now
maps  world $w$ to 
\[
    \left\{\left(\PACK{v}, \PACK{v'}\right) \isetsep \begin{array}{l}\exists \sigma, \sigma' \in \Type,
     \exists r \in \WRel{\sigma}{\sigma'},\\
     (v, v') \in \den{\Delta,\alpha\vdash \tau}\left(\extend{\phi}{\alpha}{r}\right)(w)(n)\end{array}\right\}
\]

To prove the counters are contextually equivalent we show them directly related in the
value relation. We choose the types $\sigma$ and $\sigma'$ to be $\REFTY\,\NAT$ and
the relation $r$ to be
${\lambda w . \left\{ (\ell_1, \ell_2) \isetsep \left(\ell_1, \ell_2, \left\{ \left(\underline n,
  \underline{2\cdot n}\right) \isetsep n \in \NN \right\}\right) \in w \right\}}$.
We now need to check all three functions to be related at the value relation.

First, the allocation functions. We only show one approximation, the other is completely
analogous. Concretely, we show that for any $n \in \NN$ and any world $w \in \Wl$ we have
$\left(\FUN{-}{\REF\,\underline 1}, \FUN{-}{\REF\,\underline 2}\right) \in \den{\ONE \to \alpha}(r)(w)(n)$.
Let $n \in \NN$ and $w \in \Wl$. Take $w' \geq w$ and related arguments $v, v'$ at type $\ONE$. We
know by construction that $v = v' = \unitval$ so we have to show that
$\left(\REF\,\underline 1, \REF\,\underline 2\right) \in \TT{\den{\alpha}(r)}(w')(n)$.

Let $w'' \geq w'$ and $j \leq n$ and take two related evaluation contexts $(E, E')$ at
$\T{\den{\alpha}(r)}(w'')(j)$ and $(h, h') \in \sat{w''}$. Let $\ell \not\in\dom{h}$ and
 $\ell' \not\in\dom{h'}$. We have
\[\PRk{j}{\conf{h}{E[\REF\,\underline 1]}} = \PRk{j}{\conf{\extend{h}{\ell}{\underline
        1}}{E[\ell]}}\]
and
${\PR{\conf{h'}{E'[\REF\,\underline 2]}} =
  \PR{\conf{\extend{h'}{\ell'}{\underline 2}}{E'[\ell']}}}$.

Let $w'''$ be $w''$ extended with $(\ell, \ell', r)$. Then the extended heaps are in
$\sat{w'''}$ and $w''' \geq w''$. Thus $E$ and $E'$ are also related at $w'''$ by
monotonicity. Similarly we can prove that $(\ell, \ell') \in \den{\alpha}(r)(j)(w''')$.
This then allows us to conclude ${\PRk{j}{\conf{\extend{h}{\ell}{\underline 1}}{E[\ell]}}
  \leq  \PR{\conf{\extend{h'}{\ell'}{\underline 2}}{E'[\ell']}}}$ which concludes the proof.

Lookup is simple so we omit it. Update is more interesting.
Let $n \in \NN$ and $w \in \Wl$. Let $\ell$ and $\ell'$ be related at
$\den{\alpha}(r)(w)(n)$. We need to show that
$\left(\unitval \oplus \left(\ASSIGN{\ell}{\SUCC{\LOOKUP{\ell}}}\right), \unitval \oplus
  \left(\ASSIGN{\ell'}{\SUCC({\SUCC{\LOOKUP{\ell'}}})}\right)\right) \in \TT{\den{\ONE}(r)}(w)(n)$.
Take $w' \geq w$, $j \leq n$ and $(h, h') \in \sat{w'}$. Take related evaluation contexts
$E$ and $E'$ at $w'$ and $j$. We have
\begin{align*}
  \textstyle
  \PRk{j}{\conf{h}{E\left[\unitval \oplus
  \left(\ASSIGN{\ell}{\SUCC{\LOOKUP{\ell}}}\right)\right]}} &=
  \textstyle
  \frac{1}{2}\PRk{j}{\conf{h}{E\left[\unitval\right]}} +
  \frac{1}{2}\PRk{j}{\conf{h}{E\left[\ASSIGN{\ell}{\SUCC{\LOOKUP{\ell}}}\right]}}\\
  \textstyle
  \PR{\conf{h'}{E'\left[\unitval \oplus
  \left(\ASSIGN{\ell'}{\SUCC{\SUCC{\LOOKUP{\ell'}}}}\right)\right]}} &=
  \textstyle
  \frac{1}{2}\PR{\conf{h'}{E'\left[\unitval\right]}} +
  \frac{1}{2}\PR{\conf{h'}{E'\left[\ASSIGN{\ell'}{\SUCC{\SUCC{\LOOKUP{\ell'}}}}\right]}}
\end{align*}
Since $\ell$ and $\ell'$ are related at $\den{\alpha}(r)(w)(n)$ and $w' \geq w$ and
$(h, h') \in \sat{w'}$ we know that $h(\ell) = \underline m$ and $h'(\ell') = \underline
{2 \cdot m}$ for some $m \in \NN$.

Thus
${\PRk{j}{\conf{h}{E\left[\ASSIGN{\ell}{\SUCC{\LOOKUP{\ell}}}\right]}} =
  \PRk{j}{\conf{h_1}{E[\unitval]}}}$
where $h_1 = \extend{h}{\ell}{\underline{m+1}}$. Also
${\PR{\conf{h'}{E'\left[\ASSIGN{\ell'}{\SUCC{\SUCC{\LOOKUP{\ell'}}}}\right]}} =
  \PR{\conf{h_2}{E'[\unitval]}}}$
where $h_2 = \extend{h'}{\ell'}{\underline{2\cdot(m+1)}}$. The fact that $h_1$ and $h_2$
are still related concludes the proof.

The above proof shows that reasoning about examples involving state and choice is possible
and that the two features are largely orthogonal.

\section{Conclusion}
\label{sec:conclusion}

We have constructed a step-indexed logical relation for a higher-order
language with probabilistic choice. In contrast to earlier work, our
language also features impredicative polymorphism and recursive types.
We also show how to extend our logical relation to a language with
dynamically allocated local state.  In future work, we will explore
whether the step-indexed technique can be used for developing models
of program logics for probabilistic computation that support
reasoning about more properties than just contextual equivalence.
We are also interested in including primitives for continuous
probability distributions.

\section*{Acknowledgments}
\label{sec:acknowledgments}
We thank Filip Sieczkowski, Kasper Svendsen and Thomas Dinsdale-Young for discussions of
various aspects of this work and the reviewers for their comments.

This research was supported in part by the ModuRes Sapere Aude Advanced Grant from The
Danish Council for Independent Research for the Natural Sciences (FNU) and in part by
Microsoft Research through its PhD Scholarship Programme.

\bibliographystyle{splncs03}
\bibliography{paper}

\clearpage
\appendix
\begin{center}
  {\huge \sc Appendix}
\end{center}

\section{Language definitions and properties}
\label{app:sec:lang-def-props}

\begin{figure}[htb]
  \begin{align*}
    \tau &\bnfeq
           \alpha\hmid\ONE\hmid \NAT \hmid \tau_1\times\tau_2\hmid\tau_1\plus\tau_2\hmid\tau_1\to\tau_2\hmid
           \REC\alpha\tau\hmid\ALL\alpha\tau \hmid \EX{\alpha}{\tau}\\
    v  &\bnfeq x\hmid  \unitval\hmid \underline{n} \hmid \PAIR {v_1}{v_2}\hmid \FUN xe \hmid \INL\,v \hmid \INR\,v
         \hmid \TFUN e \hmid \PACK{v}\\
    e  &\bnfeq x\hmid  \unitval\hmid \underline{n} \hmid \PAIR {e_1}{e_2}\hmid \FUN xe \hmid \INL\,e \hmid \INR\,e
         \hmid \TFUN e \hmid \PACK{e} \\ &\hspace{6ex}\hmid 
         \PROJ{i}\,e\hmid e_1\,e_2\hmid \MATCH{e}{x_1}{e_1}{x_2}{e_2}
         \hmid \TAPP e \\
         &\hspace{6ex}\hmid  \UNPACK{e_1}{x}{e_2} \hmid \UNFOLD{e} \hmid \FOLD{e} \hmid
           \RAND e \\
         &\hspace{6ex} \hmid \IFZERO{e}{e_1}{e_2} \hmid \PRED e \hmid \SUCC e\\
    E  & \bnfeq \empeval\hmid \PAIR {E}{e}\hmid \PAIR{v}{E} \hmid \INL\,E \hmid \INR\,E
         \hmid \PACK{E} \\ &\hspace{7ex}\hmid 
         \PROJ{i}\,E\hmid E\,e\hmid v\,E \hmid \MATCH{E}{x_1}{e_1}{x_2}{e_2}
         \hmid \TAPP E \\
         &\hspace{7ex}\hmid \UNPACK{E}{x}{e} \hmid \UNFOLD{E} \hmid \FOLD{E} \\
         &\hspace{7ex}\hmid \IFZERO{E}{e_1}{e_2} \hmid \RAND E\hmid \PRED E \hmid \SUCC E
  \end{align*}
  \caption{Types, terms and evaluation contexts. $\underline{n}$ are numerals of type $\NAT$.}
  \label{app:fig:syntax}
\end{figure}
\begin{figure}[htb]
  \centering
  \begin{mathpar}
    \inferrule{\alpha \in \Delta}{\Delta \vdash \alpha}
    \and
    \inferrule{}{\Delta \vdash \ONE}
    \and
    \inferrule{}{\Delta \vdash \NAT}
    \and
    \inferrule{\Delta \vdash \tau_1 \\ \Delta \vdash \tau_2}{\Delta \vdash \tau_1 \times \tau_2}
    \and
    \inferrule{\Delta \vdash \tau_1 \\ \Delta \vdash \tau_2}{\Delta \vdash \tau_1 \plus \tau_2}
    \and
    \inferrule{\Delta \vdash \tau_1 \\ \Delta \vdash \tau_2}{\Delta \vdash \tau_1 \to \tau_2}
    \and
    \inferrule{\Delta, \alpha \vdash \tau}{\Delta \vdash \EX{\alpha}{\tau}}
    \and
    \inferrule{\Delta, \alpha \vdash \tau}{\Delta \vdash \ALL{\alpha}{\tau}}
    \and
    \inferrule{\Delta, \alpha \vdash \tau}{\Delta \vdash \REC{\alpha}{\tau}}
  \end{mathpar}
  \caption{Well-formed types. The judgment $\Delta \vdash \tau$ expresses $\ftv{\tau} \subseteq \Delta$.}
  \label{app:fig:well-formed-types}
\end{figure}

\begin{figure}[htb]
\begin{mathpar}
\inferrule{x{:}\tau\in\Gamma\\ \Delta\vdash\Gamma}{\Delta\hmid\Gamma\vdash x:\tau}
\and
\inferrule{\Delta\vdash\Gamma }{\Delta\hmid\Gamma\vdash\unitval:\ONE}
\and
\inferrule{\Delta\hmid\Gamma\vdash e_1:\tau_1\\\Delta\hmid\Gamma\vdash e_2:\tau_2 }{\Delta\hmid\Gamma\vdash\PAIR{e_1}{e_2}:\tau_1\mathop\times\tau_2}
\and
\inferrule{\Delta\hmid\Gamma,x{:}\tau_1\vdash e:\tau_2\ }{\Delta\hmid\Gamma\vdash\FUN xe:\tau_1\mathop\to\tau_2}
\and
\inferrule{\Delta\hmid\Gamma\vdash e:\tau_1 \\ \Delta \vdash \tau_2}{\Delta\hmid\Gamma\vdash\INL\,e:\tau_1\plus \tau_2}
\and
\inferrule{\Delta\hmid\Gamma\vdash e:\tau_2 \\ \Delta \vdash \tau_1}{\Delta\hmid\Gamma\vdash\INR\,e:\tau_1\plus \tau_2}
\and
\inferrule{\Delta\hmid\Gamma, x_1 {:} \tau_1 \vdash e_1 : \tau \\ \Delta\hmid\Gamma, x_2 {:} \tau_2 \vdash e_2 : \tau\\
\Delta\hmid\Gamma \vdash e : \tau_1 \plus \tau_2}
{\Delta \hmid\Gamma \vdash \MATCH{e}{x_1}{e_1}{x_2}{e_2} : \tau}\and
\inferrule{\Delta,\alpha\hmid\Gamma\vdash e:\tau}{\Delta\hmid\Gamma\vdash\TFUN{e} : \ALL\alpha\tau}
\and
\inferrule{\Delta\hmid\Gamma\vdash e:\tau_1\times\tau_2}{\Delta\hmid\Gamma\vdash \PROJ{i}\,e:\tau_i}
\and
\inferrule{\Delta\hmid\Gamma\vdash e:\tau'\to\tau\\ \Delta\hmid\Gamma\vdash e':\tau'}{\Delta\hmid\Gamma\vdash e\,e':\tau}
\and
\inferrule{\Delta \vdash \tau_1 \\ \Delta\hmid\Gamma \vdash e : \tau\subst{\alpha}{\tau_1}}
          {\Delta\hmid\Gamma \vdash \PACK{e} : \EX{\alpha}{\tau}}
\and
\inferrule{\Delta\hmid\Gamma \vdash e : \EX{\alpha}{\tau_1} \\ \Delta \vdash \tau \\ \Delta, \alpha\hmid \Gamma, x : \tau_1 \vdash e' : \tau}
          {\Delta\hmid\Gamma \vdash \UNPACK{e}{x}{e'} : \tau}
\and
\inferrule{\Delta\hmid \Gamma \vdash e : \REC{\alpha}{\tau}}
          {\Delta\hmid \Gamma \vdash \UNFOLD{e} : \tau\subst{\alpha}{\REC{\alpha}{\tau}}}
\and
\inferrule{\Delta\hmid \Gamma \vdash e : \tau\subst{\alpha}{\REC{\alpha}{\tau}}}
          {\Delta\hmid \Gamma \vdash \FOLD{e} : \REC{\alpha}{\tau}}
\and
\inferrule{\Delta\hmid\Gamma\vdash e:\ALL\alpha\tau\\\Delta\vdash\tau'}{\Delta\hmid\Gamma\vdash\TAPP{e}:\tau\subst\alpha{\tau'}}
\and
{\color{red}{\inferrule{\Delta\hmid \Gamma \vdash e : \NAT}{\Delta\hmid\Gamma\vdash \RAND\,e : \NAT}}}\and
\inferrule{\Delta\hmid\Gamma\vdash e:\NAT\\ \Delta\hmid\Gamma\vdash e_1 : \tau \\
  \Delta\hmid\Gamma\vdash e_2 : \tau}{\Delta\hmid\Gamma\vdash\IFZERO{e}{e_1}{e_2}:\tau}
\and
\inferrule{\Delta\hmid \Gamma \vdash e : \NAT}{\Delta\hmid\Gamma\vdash \PRED e : \NAT}
\and
\inferrule{\Delta\hmid \Gamma \vdash e : \NAT}{\Delta\hmid\Gamma\vdash \SUCC e : \NAT}
\end{mathpar}
\caption{Typing of terms, where $\Gamma\bnfeq\emptyset\hmid\Gamma,x{:}\tau$ and $\Delta\bnfeq\emptyset\hmid\Delta,\alpha$.}
\label{app:fig:typing-full}
\end{figure}

\begin{figure*}
  Basic reductions $\basicstepsto{\cdot}$
  \begin{align*}
    \PROJ{i}\,\PAIR{v_1}{v_2} &\basicstepsto{1} v_i
    & 
      \UNFOLD{(\FOLD{v})} &\basicstepsto{1} v\\
    (\FUN x e)\,v &\basicstepsto{1} e \subst x v 
    &
      \UNPACK{(\PACK{v})}{x}{e} &\basicstepsto{1} e\subst{x}{v}\\
    \TAPP{(\TFUN{e})} &\basicstepsto{1} e
    & \MATCH{\INL{v}}{x_1}{e_1}{x_2}{e_2} &\basicstepsto{1} e_1\subst{x_1}{v}\\
    \color{red}\RAND\underline{n} &\color{red}\basicstepsto{\frac{1}{n}} \underline k \quad\ \  \left(k \in \{1,2,\ldots,n\}\right)
    &\MATCH{\INR{v}}{x_1}{e_1}{x_2}{e_2} &\basicstepsto{1} e_2\subst{x_2}{v}\\
    \PRED\underline{n} &\basicstepsto{1} \underline{\max\{n-1,1\}}
    &\SUCC\underline{n} &\basicstepsto{1} \underline{n+1}\\
    \IFZERO{\underline 1}{e_1}{e_2}&\basicstepsto{1} e_1 &
    \IFZERO{\SUCC \underline n}{e_1}{e_2}&\basicstepsto{1} e_2\\
    \intertext{One step reduction relation $\stepsto{\cdot}$}
    E[e]&\stepsto{p} E[e'] \qquad\qquad\text{if $e\basicstepsto{p} e'$}
  \end{align*}
  \caption{Operational semantics.}
  \label{app:fig:semantics}
\end{figure*}
The following lemma uses definitions from Section~\ref{sec:term-relat}.
\begin{lemma}
  \label{app:lem:prob-convergence-cont}
  $\Phi$ is monotone and preserves suprema of $\omega$-chains.
\end{lemma}
\begin{proof}
  Since the order in $\Fl$ is pointwise and multiplication and addition are monotone it is
  easy to see that $\Phi$ is monotone.

  To show that it is continuous let $\{f_n\}_{n\in\omega}$ be an $\omega$-chain in
  $\Fl$. If $e$ is a value the result is immediate. Otherwise we have
  \begin{align*}
    \Phi\left(\sup_{n\in\omega} f_n\right)(e) &=
      \displaystyle\sum_{e \stepsto{p}{e'}} p \cdot \left(\sup_{n\in\omega} f_n\right)\left(e'\right)
    \intertext{and since suprema in $\Fl$ are computed pointwise we have}
        &=
            \displaystyle\sum_{e \stepsto{p}{e'}} p \cdot \sup_{n\in\omega} \left(f_n(e')\right)
  \end{align*}
  Using the fact that sum and product are continuous and that the sum in the definition of
  $\Phi$ is finite we get
  \begin{align*}
    \Phi\left(\sup_{n\in\omega} f_n\right)(e) &=
      \displaystyle\sup_{n\in\omega}\left(\sum_{e \stepsto{p}{e'}} p \cdot f_n(e')\right)\\
    &= \sup_{n\in\omega}\Phi\left(f_n\right)(e) = \left(\sup_{n\in\omega}\Phi(f_n)\right)\left(e\right)
  \end{align*}
\end{proof}

\begin{example}
  \label{app:ex:prob-of-convergence}
  Let us compute probabilities of termination of some example programs.

  \begin{itemize}
  \item If $v \in \Vals$ then by definition $\PR{v} = 1$.
  \item If $e \in \Exps \setminus \Vals$ is stuck then $\PR{e} = 0$ by definition.
  \item Suppose there exists a cycle $e \stepsto{1} e_1 \stepsto{1} e_2 \stepsto{1} \cdots
    \stepsto{1} e_n \stepsto{1} e$. Then $\PR{e} = \PR{e_1} = \cdots = \PR{e_n} = 0$.

    It follows from the assumption that none of $e_k$ are values and since the sum of
    outgoing weights is at most $1$ we have that for each $e_k$ and $e$ all other weights
    must be $0$. We thus get that $\PR{e} = \PR{e_1} = \cdots = \PR{e_n}$ by simply
    unfolding the fixed point $n$-times. To show that they are all $0$ we use Scott
    induction. Define
    \[\Sl = \left\{ f \in \Fl \middle| f(e) = f(e_1) = f(e_2) = \ldots = f(e_n)
    = 0 \right\}.
    \]
    Clearly $\Sl$ is an admissible subset of $\Fl$ and $\bot \in \Sl$.  Using the above
    existence of the cycle of reductions it is easy to show that $\Sl \subseteq
    \Phi\left[\Sl\right]$. Hence by the principle of Scott induction we have $\PR{\cdot}
    \in \Sl$ and thus $\PR{e} = \PR{e_1} = \ldots = \PR{e_n} = 0$.
  \end{itemize}
\end{example}
This example also shows that we do really want the least fixed point of $\Phi$, since this
allows us to use Scott-induction and prove that diverging terms have zero probability of
termination.

\begin{remark}
  \label{app:rem:probs-of-convergence}
  It is perhaps instructive to consider the relationship to the termination predicate when we do
  not have weights on reductions. In such a case we can consider two extremes, may- and
  must-termination predicates. These can be considered to be maps $\Exps \to \mathbf{2}$
  where $\mathbf{2}$ is the boolean lattice $0 \leq 1$. Let $\Bl = \Exps \to \mathbf{2}$. Since
  $\mathbf{2}$ is a complete lattice so is $\Bl$. In particular it is a pointed
  $\omega$-cpo. We can define may-termination as the least fixed point of $\Psi : \Bl \to
  \Bl$ defined as
  \begin{align*}
    \Psi(f)(e) =
    \begin{cases}
      1 &\text{if } e \in \Vals\\
      \displaystyle \max_{e \leadsto e'}f(e') &\text{otherwise}
    \end{cases}.
  \end{align*}
  Observe again that if $e$ is stuck then $\Psi(f)(e) = 0$ since the maximum of an empty
  set is the least element by definition.

  Must-termination is slightly different. We need a special case for stuck terms.
  \begin{align*}
    \Psi'(f)(e) =
    \begin{cases}
      1 &\text{if } e \in \Vals\\
      \displaystyle \min_{e \leadsto e'}f(e') &\exists e' \in \Exps, p \in \interval, e \stepsto{p}{e'}\\
      0 & \text{otherwise}
    \end{cases}
  \end{align*}

  Let $\downarrow$ be the least fixed point of $\Psi$ and $\Downarrow$ the least fixed
  point of $\Psi'$. An additional property that holds for $\downarrow$ and $\Downarrow$,
  because of the fact that $\mathbf{2}$ is discrete, is
  that for a given $e$, if $e\hspace{-3pt}\downarrow\, = 1$ then there is a natural number $n$, such that
  $\Psi^n(\bot)(e) = 1$, i.e. if it terminates we can observe this in finite
  time. This is because if an increasing sequence in $\mathbf{2}$ has supremum $1$, then
  the sequence must be constant $1$ from some point onward.

  In contrast, if $\PR{e} = 1$ it is not necessarily the case that there is a natural
  number $n$ with $\Phi^n(\bot)(e) = 1$ because it might be the case that $1$ is
  only reached in the limit.
\end{remark}

The next lemma uses the abbreviation $;$ defined in Section~\ref{sec:inhabitants-forall}.
\begin{lemma}
  \label{app:lem:sequence-product}
  For all terms $e, e' \in \Exps$, $\PR{e; e'} = \PR{e} \cdot \PR{e'}$.
\end{lemma}
\begin{proof}
  We prove two approximations separately, both of them by Scott induction.
  \begin{description}
  \item[$\leq$] Consider the set
    \begin{align*}
      \Sl = \left\{f \in \Fl \ \middle|\ \begin{array}{l}
          f \leq \PR{\cdot} \land \forall e, e' \in \Exps,\\
          f(e ; e') \leq \PR{e} \cdot \PR{e'} \end{array} \right\}.
    \end{align*}
    It is easy to see that $\Sl$ contains $\bot$ and is closed under $\omega$-chains, so we only need to
    show that it is preserved by $\Phi$. The first condition is trivial to check since
    $\PR{\cdot}$ is a fixed point of $\Phi$. Let $f \in \Fl$ and $e, e' \in \Exps$. If $e
    \in \Vals$ then $\Phi(f)(e ; e') = f(e')$ on account of one $\beta$-reduction. By
    assumption $f(e') \leq \PR{e'}$ and by definition we have $\PR{e} = 1$.

    If $e$ is not a value we have
    $\Phi(f)(e; e') = \sum_{e \stepsto{p}{e''}} p \cdot f\left(e'' ; e'\right) \leq
    \sum_{e \stepsto{p}{e''}} p \cdot \PR{e''}\cdot\PR{e'} = \PR{e}\cdot\PR{e'}$.

    Thus we can conclude by Scott induction that $\PR{\cdot} \in \Sl$.
  \item[$\geq$] For this direction we consider the set
    \begin{align*}
      \Sl = \left\{f \in \Fl \ \middle|\ 
        \begin{array}{l}
          \forall E \in \Ectxts, e \in \Exps, v \in \Vals,\\
          \PR{E[e]} \geq f(e) \cdot \PR{E[v]}\end{array} \right\}.
    \end{align*}
    It is easy to see that it is admissible and closed under $\Phi$. Hence $\PR{\cdot} \in
    \Sl$.  Thus we have, taking $E = - ; e'$ and any value $v$, that $\PR{e} \cdot \PR{v ;
      e'} \leq \PR{e; e'}$ and it is easy to see that $\PR{v ; e'} = \PR{e'}$.
  \end{description}
\end{proof}

\begin{lemma}
  \label{app:lem:pure-reductions-dont-change}
  Let $e, e' \in \Exps$. If $e \stepstozeropure e'$ then for all $k$, $\PRk{k}{e} =
  \PRk{k}{e'}$.
\end{lemma}
\begin{proof}
  When $k$ is $0$ the result is immediate. So assume $k > 0$. We need to distinguish two
  cases.
  \begin{itemize}
  \item If there exists $v' \in \Vals$ such that $e' \stepstozeropure v'$ then
    we also have $e \stepstozeropure v'$ and we are done.
  \item If not, then we need to inspect the definition of $\reds{e}$ and $\reds{e'}$. It
    is easy to see that any path $\pi \in \reds{e'}$ corresponds to a unique path
    $\pi'\cdot\pi$ in $\reds{e}$. It is similarly easy to see that $\weight{\pi} =
    \weight{\pi'\cdot\pi}$ and that $\last{\pi} = \last{\pi'\cdot\pi}$. Thus we have
    that $\PRk{k}{e} = \PRk{k}{e'}$.
  \end{itemize}
\end{proof}

\begin{proposition}
  \label{app:prop:prob-smaller}
  For each $e \in \Exps$ we have $\PR{e} \leq \sup_{k\in\omega}\left(\PRk{k}{e}\right)$.
\end{proposition}
\begin{proof}
  We use Scott induction. Let $\Sl$ be the set
  \[
  \Sl = \left\{ f \in \Fl \ \middle|\ \forall e, f(e) \leq \sup_{k\in\omega}\left(\PRk{k}{e}\right) \right\}
  \]
  It is easy to see that $\Sl$ is closed under limits of $\omega$-chains and that $\bot
  \in \Sl$ so we only need to show that $\Sl$ is closed under $\Phi$. Let $f \in \Sl$ and
  $e$ an expression. We have
  \begin{align*}
    \Phi(f)(e) =
    \begin{cases}
      1 &\text{if } e \in \Vals\\
      \displaystyle\sum_{e \stepsto{p}{e'}} p \cdot f\left(e'\right) &\text{otherwise}
    \end{cases}
  \end{align*}
  and we consider $4$ cases.
  \begin{itemize}
  \item $e\in\Vals$. We always have $e \stepstozeropure e$ and so we have that for any $k > 0$,
    $\PRk{k}{e} = 1$ which is the top element.
  \item $e\stepsto{p} e'$ and the reduction is not unfold-fold or choice.  Then we use
    Lemma~\ref{lem:pure-reductions-dont-change} to get $\PRk{k}{e} = \PRk{k}{e'}$ for all $k$.
    Similarly we have that $\Phi(f)(e) = f(e')$ from the
    definition of $\Phi$. Thus we can use the assumption that $f\in\Sl$.
  \item $e\stepsto{1} e'$ and the reduction is unfold-fold. This follows directly from the
    definition of $\reds{\cdot}$, $\Psi$ and the assumption that $f\in\Sl$.
  \item The reduction from $e$ is a choice reduction. Suppose $e$ reduces to $e_1, e_2,
    \ldots, e_n$. Then we know from the operational semantics that the weights are all
    $\tfrac{1}{n}$. We get
    \begin{align}
      \label{app:eq:phi-expand}
      \Phi(f)(e) = \sum_{i=1}^n\frac{1}{n} f(e_i)
      \qquad \text{ and }\qquad
      \PRk{k+1}{e} = \sum_{i=1}^n\frac{1}{n} \PRk{k}{e_i}.
    \end{align}
    Using the fact that $\PRk{k}{e_i}$ is an increasing chain in $k$ for
    each $e_i$ we have
    \begin{align}
      \label{app:eq:sup-expand}
    \sup_{k\in\omega}\left(\PRk{k}{e}\right) &= 
                                  \sum_{i=1}^n\frac{1}{n} \sup_{k\in\omega}\left(\PRk{k}{e_i}\right)
    \end{align}
    By assumption $f(e_i) \leq \sup_{k\in\omega}\left(\PRk{k}{e_i}\right)$ for all $i \in
    \{1, 2, \ldots, n\}$ which concludes the proof using \eqref{app:eq:phi-expand} and
    \eqref{app:eq:sup-expand}.
  \end{itemize}
\end{proof}

\subsubsection{Interpretation of types and the logical relation}

\begin{figure}[htb]
\begin{mathpar}
\inferrule{x{:}\tau\in\Gamma}{\Delta\hmid\Gamma\vdash x\mathrel\R x:\tau}
\and
\inferrule{  }{\Delta\hmid\Gamma\vdash\unitval \mathrel\R \unitval:\ONE}
\and
\inferrule{\Delta\hmid\Gamma\vdash e_1\mathrel\R e_1':\tau_1\\\Delta\hmid\Gamma\vdash e_2\mathrel\R e_2':\tau_2 }{\Delta\hmid\Gamma\vdash\PAIR{e_1}{e_2}\mathrel\R \PAIR{e_1'}{e_2'}:\tau_1\times\tau_2}
\and
\inferrule{\Delta\hmid\Gamma,x{:}\tau_1\vdash e\mathrel\R e':\tau_2 }{\Delta\hmid\Gamma\vdash\FUN xe\mathrel\R\FUN x{e'}:\tau_1\to\tau_2}
\and
\inferrule{\Delta\hmid\Gamma\vdash e\mathrel\R e':\tau_1}{\Delta\hmid\Gamma\vdash\INL{e}\mathrel\R \INL{e'}:\tau_1\plus \tau_2}
\and
\inferrule{\Delta\hmid\Gamma\vdash e\mathrel\R e':\tau_2}{\Delta\hmid\Gamma\vdash\INR{e}\mathrel\R \INR{e'}:\tau_1\plus \tau_2}\and
\inferrule{\Delta\hmid\Gamma, x_1 {:} \tau_1 \vdash e_1 \mathrel\R e_1': \tau \\ 
           \Delta\hmid\Gamma, x_2 {:} \tau_2 \vdash e_2 \mathrel\R e_2' : \tau\\
\Delta\hmid\Gamma \vdash e \mathrel\R e' : \tau_1 \plus \tau_2}
{\Delta\hmid\Gamma \vdash \MATCH{e}{x_1}{e_1}{x_2}{e_2}\mathrel\R \MATCH{e'}{x_1}{e_1'}{x_2}{e_2'} : \tau}\and
\inferrule{\Delta,\alpha\hmid\Gamma\vdash e\mathrel\R e':\tau}{\Delta\hmid\Gamma\vdash\TFUN{e}\mathrel\R\TFUN{e'} : \ALL\alpha\tau}
\and
\inferrule{\Delta \vdash \tau_1 \\ \Delta\hmid\Gamma \vdash e \mathrel\R e' : \tau\subst{\alpha}{\tau_1}}
          {\Delta\hmid\Gamma \vdash (\PACK{e})
            \mathrel\R (\PACK{e'}) : \EX{\alpha}{\tau}}
\and
\inferrule{\Delta\hmid\Gamma \vdash e_1 \mathrel\R e_1': \EX{\alpha}{\tau_1} \\ \Delta \vdash \tau
  \\ \Delta, \alpha\hmid \Gamma, x : \tau_1 \vdash e \mathrel\R e': \tau}
          {\Delta\hmid\Gamma \vdash (\UNPACK{e_1}{x}{e}) \mathrel\R (\UNPACK{e_1'}{x}{e'}) : \tau}
\and
\inferrule{\Delta\hmid\Gamma\vdash e\mathrel\R e':\tau_1\times\tau_2}{\Delta\hmid\Gamma\vdash \PROJ{i}\,e\mathrel\R\PROJ{i}\,e':\tau_i}
\and
\inferrule{\Delta\hmid\Gamma\vdash e_1\mathrel\R e_1':\tau'\to\tau\\ \Delta\hmid\Gamma\vdash e_2\mathrel\R e_2':\tau'}{\Delta\hmid\Gamma\vdash e_1\,e_2\mathrel\R e_1'\,e_2':\tau}
\and
\inferrule{\Delta\hmid \Gamma \vdash e \mathrel\R e' : \REC{\alpha}{\tau}}
          {\Delta\hmid \Gamma \vdash \UNFOLD{e}\mathrel\R \UNFOLD{e'}: \tau\subst{\alpha}{\REC{\alpha}{\tau}}}
\and
\inferrule{\Delta\hmid \Gamma \vdash e\mathrel\R e' : \tau\subst{\alpha}{\REC{\alpha}{\tau}}}
          {\Delta\hmid \Gamma \vdash \FOLD{e}\mathrel\R \FOLD{e'} : \REC{\alpha}{\tau}}
\and
\inferrule{\Delta\hmid\Gamma\vdash e\mathrel\R e':\ALL\alpha\tau}{\Delta\hmid\Gamma\vdash\TAPP{e}\mathrel\R\TAPP{e'}:\tau\subst\alpha{\tau'}}\ \ftv{\tau'}\subseteq\Delta
\and
\inferrule{\Delta\hmid\Gamma\vdash e \mathrel\R e': \NAT }{\Delta\hmid\Gamma\vdash \RAND
  e\mathrel\R\RAND e' :\NAT}\\
\inferrule{\Delta\hmid\Gamma\vdash e \mathrel\R e': \NAT }{\Delta\hmid\Gamma\vdash \PRED
  e\mathrel\R\PRED e' :\NAT}
\and
\inferrule{\Delta\hmid\Gamma\vdash e \mathrel\R e': \NAT }{\Delta\hmid\Gamma\vdash \SUCC
  e\mathrel\R\SUCC e' :\NAT}\\
\inferrule{\Delta\hmid\Gamma \vdash e \mathrel\R e' : \NAT\\
           \Delta\hmid\Gamma,  \vdash e_1 \mathrel\R e_1': \tau \\ 
           \Delta\hmid\Gamma,  \vdash e_2 \mathrel\R e_2' : \tau}
{\Delta \hmid\Gamma \vdash \IFZERO{e}{e_1}{e_2}\mathrel\R \IFZERO{e'}{e_1'}{e_2'} : \tau}
\end{mathpar}
\caption{Compatibility properties of type-indexed relations}
\label{app:fig:compatibility}
\end{figure}

\begin{figure}[htb]
  \centering
  \begin{align*}
    \den{\Delta \vdash \alpha}(\phi) &= \phi_r(\alpha)\\
    \den{\Delta \vdash \NAT}(\phi)(n) &= \left\{(\underline k, \underline k) \ \middle|\  k
                                        \in \NN, k > 0\right\} \\
    \den{\Delta \vdash \tau \times \sigma}(\phi)(n) &= 
    \left\{\left(\PAIR{v}{u},
        \PAIR{v'}{u'}\right) \ \middle|\  \begin{array}{l}
        (v, v') \in \den{\Delta \vdash \tau}(\phi)(n),\\
        (u, u') \in \den{\Delta\vdash \sigma}(\phi)(n)\end{array}\right\} \\
    \den{\Delta \vdash \tau \plus \sigma}(\phi)(n) &= 
    \left\{\left(\INL\,v, \INL\,v'\right) \ \middle|\  (v, v') \in \den{\Delta \vdash
                                                     \tau}(\phi)(n)\right\}\\ &\quad\ \cup
    \left\{\left(\INR\,v, \INR\,v'\right) \ \middle|\  (v, v') \in \den{\Delta \vdash
                                                     \sigma}(\phi)(n)\right\} \\
    \den{\Delta \vdash \tau \to \sigma}(\phi)(n) &= \left\{\left(\FUN{x}{e},
    \FUN{y}{e'}\right) \ \middle|\  \begin{array}{l}\forall j \leq n, 
     \forall (v, v') \in \den{\Delta \vdash \tau}(\phi)(j),\\
     ((\FUN{x}{e})\,v, (\FUN{y}{e'})\,v') \in \TT{\den{\Delta\vdash \sigma}(\phi)}(j)\end{array}\right\} \\
    \den{\Delta \vdash \ALL{\alpha}{\tau}}(\phi)(n) &= \left\{\left(\TFUN{e},
    \TFUN{e'}\right) \ \middle|\ \begin{array}{l}\forall \sigma, \sigma' \in \Type,
     \forall r \in \VRel{\sigma}{\sigma'},\\
     (e, e') \in \TT{\den{\Delta,\alpha\vdash
         \tau}\left(\extend{\phi}{\alpha}{r}\right)}(n)\end{array}\right\}\\
    \den{\Delta \vdash \EX{\alpha}{\tau}}(\phi)(n) &= \left\{\left(\PACK{v},
    \PACK{v'}\right) \ \middle|\  \begin{array}{l}\exists \sigma, \sigma' \in \Type,
     \exists r \in \VRel{\sigma}{\sigma'},\\
     (v, v') \in \den{\Delta,\alpha\vdash \tau}\left(\extend{\phi}{\alpha}{r}\right)(n)
      \end{array}\right\}\\
    \den{\Delta \vdash \REC{\alpha}{\tau}}(\phi)(0) &=
         \Val{\phi_1(\REC{\alpha}{\tau})}\times \Val{\phi_2(\REC{\alpha}{\tau})}\\
    \den{\Delta \vdash \REC{\alpha}{\tau}}(\phi)(n+1) &= \left\{\left(\FOLD{v},
    \FOLD{v'}\right) \ \middle|\   
     (v, v') \in \den{\Delta,\alpha\vdash \tau}\left(\extend{\phi}{\alpha}{\den{\Delta\vdash \REC{\alpha}{\tau}}(\phi)}\right)(n)\right\}
 \end{align*}
  \caption{Interpretation of types.}
  \label{app:fig:logical-relation}
\end{figure}

\begin{lemma}
  \label{app:lem:well-defined}
  The interpretation of types in Fig.~\ref{fig:logical-relation} is well defined. In
  particular the interpretation of types is non-expansive.
\end{lemma}

The substitution lemma is crucial for proving compatibility of existential and universal
types. The proof is by induction.
\begin{lemma}[Substitution]
  \label{app:lem:interpretation-subst}
  For any well-formed types $\Delta,\alpha \vdash \tau$ and $\Delta \vdash \sigma$ and any
  $\phi$ we have 
  $\den{\Delta \vdash \tau\subst\alpha\sigma}(\phi) = \den{\Delta, \alpha \vdash
    \tau}\left(\extend{\phi}{\alpha}{\den{\Delta \vdash \sigma}(\phi)}\right)$.
\end{lemma}

We state and prove additional context extension lemmas. The other cases are similar.
\begin{lemma}
  \label{app:lem:ctx-extend-func}
  Let $n \in \NN$.
  If $(v, v') \in \den{\Delta \vdash \tau_1 \to \tau_2}(\phi)(n)$ and $(E, E') \in
  \T{\den{\Delta \vdash \tau_2}(\phi)}(n)$ then 
  $\left(\ecomp{E}{(v\,[])}, \ecomp{E'}{(v'\,[])}\right) \in \T{\den{\Delta \vdash \tau_1}(\phi)}(n)$.
\end{lemma}
This follows directly from the definition of the interpretation of types.

\begin{corollary}
  \label{app:cor:ctx-extend-app}
  Let $n\in\NN$.
  If $(e, e') \in \TT{\den{\Delta \vdash \tau_1}(\phi)}(n)$ and $(E, E') \in
  \T{\den{\Delta \vdash \tau_2}(\phi)}(n)$ then 
  \[
  \left(\ecomp{E}{([]\,e)}, \ecomp{E'}{([]\,e')}\right) \in \T{\den{\Delta \vdash \tau_1 \to \tau_2}(\phi)}(n).\]
\end{corollary}
\begin{proof}
  Let $n\in\NN$.
  Take $(v, v') \in\den{\Delta \vdash \tau_1\to\tau_2}(\phi)(n)$. By Lemma~\ref{app:lem:ctx-extend-func} 
  and monotonicity we have for all $k\leq n$, $(\ecomp{E}{(v\,[])}, \ecomp{E'}{(v'\,[])}) \in \T{\den{\Delta \vdash \tau_1}(\phi)}(k)$
  and by the assumption that
  $(e, e') \in \TT{\den{\Delta \vdash \tau_1}(\phi)}(n)$ we have
  \begin{align*}
    \PRk{k}{E[v\,e]} \leq \PR{E'[v'\,e']}
  \end{align*}
  concluding the proof.
\end{proof}

\begin{lemma}
  \label{app:lem:ctx-extend-unfold}
  Let $n\in\NN$.
  If $(E, E') \in \T{\den{\Delta \vdash \tau\subst{\alpha}{\REC{\alpha}{\tau}}}(\phi)}(n)$ then 
  \[
  \left(\ecomp{E}{(\UNFOLD{[]})}, \ecomp{E'}{(\UNFOLD{[]})}\right) \in 
  \T{\den{\Delta \vdash \REC{\alpha}{\tau}}(\phi)}(n).\]
\end{lemma}
\begin{proof}
  Let $n\in\NN$. We consider two cases.

  \begin{itemize}
  \item $n = m + 1$

    Take $(\FOLD{v}, \FOLD{v'}) \in
    \den{\Delta\vdash\REC{\alpha}{\tau}}(\phi)(n)$. By definition
    \[(v,v') \in \den{\Delta
      \vdash \tau\subst{\alpha}{\REC{\alpha}{\tau}}}(\phi)(m).\]  Let $k \leq n$. If $k=0$
    the condition is trivially true (since $\PRk{k}{E[\UNFOLD{\FOLD{v}}]} = 0$) so assume
    $k = \ell + 1$. Note that crucially $\ell \leq m$. Using
    Lemma~\ref{lem:tt-closure},
    Lemma~\ref{lem:prob-unfold-fold-choice} and
    Lemma~\ref{lem:choice-free-reductions-dont-change} we have
    \begin{align*}
      \PRk{k}{E[\UNFOLD{(\FOLD{v})}]} &= \PRk{\ell}{E[v]} \\
                                      &\leq \PR{E'[v']}\\
                                      &= \PR{E'[\UNFOLD{(\FOLD{v'})}]}
    \end{align*}
    concluding the proof.
  \item $n = 0$. This case is trivial, since $\PRk{0}{e} = 0$ for any $e$.
  \end{itemize}
\end{proof}

\begin{lemma}
  \label{app:lem:ctx-extend-fold}
  Let $n\in\NN$.
  If $(E, E') \in   \T{\den{\Delta \vdash \REC{\alpha}{\tau}}(\phi)}(n)$ then 
  \[
  \left(\ecomp{E}{(\FOLD{[]})}, \ecomp{E'}{(\FOLD{[]})}\right) \in 
  \T{\den{\Delta \vdash \tau\subst{\alpha}{\REC{\alpha}{\tau}}}(\phi)}(n).\]
\end{lemma}
\begin{proof}
  Easily follows from the fact that if $(v, v')$ are related at the unfolded type then
  $(\FOLD{v}, \FOLD{v'})$ are related at the folded type (using weakening to get to the
  same stage).
\end{proof}

To relate the logical relation to contextual and CIU approximations we first have that the
composition of logical and CIU approximations is included in the logical approximation relation.
\begin{corollary}
  \label{app:cor:logrel-ciu-extend}
  If $\logapprox{\Delta}{\Gamma}{e}{e'}{\tau}$ and
  $\ciuapprox{\Delta}{\Gamma}{e'}{e''}{\tau}$ then 
  $\logapprox{\Delta}{\Gamma}{e}{e''}{\tau}$.
\end{corollary}
This follows directly from the definition. This corollary in turn implies, together with
Proposition~\ref{prop:logrel-is-compatible} and the fact that all compatible relations are
in particular reflexive, that CIU approximation relation is contained in the logical
relation.
\begin{corollary}
  \label{app:cor:ciu-implies-logrel}
  If $\ciuapprox{\Delta}{\Gamma}{e}{e'}{\tau}$ then
  ${\logapprox{\Delta}{\Gamma}{e}{e'}{\tau}}$.
\end{corollary}

Finally we have adequacy of the logical relation.
\begin{corollary}
  \label{app:cor:logrel-is-adequat}
  Logical approximation relation $\logapproxrel$ is adequate.
\end{corollary}
\begin{proof}
  Assume $\logapprox{\emp}{\emp}{e}{e'}{\tau}$. We are to show that $\PR{e} \leq
  \PR{e'}$. Straight from the definition we have $\forall n \in \NN, (e, e') \in \TT{\den{\emp \vdash \tau}}(n).$
  The empty evaluation context is always related to itself (at any
  type). This implies $\forall n \in \NN, \PRk{n}{e} \leq \PR{e'}$
  which further implies (since the right-hand side is independent of $n$) that 
  $\sup_{n\in\omega}\left(\PRk{n}{e}\right) \leq \PR{e'}$.
  Using Proposition~\ref{prop:prob-smaller} we thus have
  $\PR{e} \leq \sup_{n\in\omega}\left(\PRk{n}{e}\right) \leq \PR{e'}$
  concluding the proof.
\end{proof}

\begin{lemma}[Functional extensionality for values]
  \label{app:lem:func-extensionality}
  Suppose $\tau, \sigma \in \Type(\Delta)$ and let $\FUN{x}{e}$ and $\FUN{x'}{e'}$ be two
  values of type $\tau\to\sigma$ in context $\Delta\ |\ \Gamma$.
  If for all $u \in \Val{\tau}$ we have
  $\ctxapprox{\Delta}{\Gamma}{\left(\FUN{x}{e}\right)\,u}{\left(\FUN{x'}{e'}\right)\,u}{\sigma}$
  then
  \[{\ctxapprox{\Delta}{\Gamma}{\FUN{x}{e}}{\FUN{x'}{e'}}{\tau\to\sigma}}\enspace.\]
\end{lemma}
\begin{proof}
  We use Theorem~\ref{thm:CIU-theorem} several times and show $\FUN{x}{e}$ and
  $\FUN{x'}{e'}$ are logically related. Let $n \in \NN$, $\phi \in \VRelD{\Delta}$ and
  $(\gamma, \gamma') \in \den{\Delta \vdash \Gamma}(\phi)(n)$. Let $v = \FUN{x}{e\gamma}$
  and $v' = \FUN{x'}{e'\gamma'}$. We are to show $(v, v') \in \TT{\den{\Delta \vdash
      \tau\to\sigma}(\phi)}(n)$ and to do this we show directly $(v, v') \in \den{\Delta
    \vdash \tau\to\sigma}(\phi)(n)$.

  Let $j \leq n$, $(u, u') \in \den{\tau}(\phi)(n)$, $k \leq j$ and $(E, E') \in
  \T{\den{\sigma}(\phi)}(k)$. We have to show
  $\PRk{k}{E[v\,u]} \leq \PR{E'[v'\,u']}$. From
  Proposition~\ref{prop:logrel-is-compatible} we have that $(v, v) \in
  \TT{\den{\tau\to\sigma}(\phi)}(n)$ and so
  $\PRk{k}{E[v\,u]} \leq \PR{E'[v\,u']}$. From the assumption of the lemma we have that
  $v\,u' \ciuapproxrel v'\,u'$ which concludes the proof.
\end{proof}

\begin{lemma}[Extensionality for the universal type]
  \label{app:lem:all-extensionality}
  Let $\tau \in \Type(\Delta, \alpha)$ be a type.
  Let $\TFUN{e}, \TFUN{e'}$ be two terms of type $\ALL{\alpha}{\tau}$ in context $\Delta\ |\
  \Gamma$. If for all closed types $\sigma \in \Type$ we have
  \begin{align*}
    \ctxapprox{\Delta}{\Gamma}{e}{e'}{\tau\subst{\alpha}{\sigma}}
  \end{align*}
  then $\ctxapprox{\Delta}{\Gamma}{\TFUN{e}}{\TFUN{e'}}{\ALL{\alpha}{\tau}}$.
\end{lemma}
\begin{proof}
  We again use Theorem~\ref{thm:CIU-theorem} multiple times. Let $n \in \NN$, $\phi \in
  \VRelD{\Delta}$ and $(\gamma, \gamma') \in \den{\Delta\vdash\Gamma}(\phi)(n)$. Let
  $v = \TFUN{e\gamma}$ and $v' = \TFUN{e'\gamma'}$. We show directly that $(v, v') \in
  \den{\Delta \vdash \ALL{\alpha}{\tau}}(\phi)(n)$.

  So take $\sigma, \sigma' \in \Type$ and $r \in \VRel{\sigma}{\sigma'}$ and we need to show
  $(e\gamma, e'\gamma') \in \TT{\den{\Delta, \alpha}(\extend{\phi}{\alpha}{r})}(n)$. Let $k \leq n$ and
  $(E, E')$ related at $k$. We have to show $\PRk{k}{E[e\gamma]} \leq \PR{E'[e'\gamma']}$.
  From Proposition~\ref{prop:logrel-is-compatible} we have
  \[(e\gamma, e\gamma') \in \TT{\den{\Delta, \alpha}(\extend{\phi}{\alpha}{r})}(n)\] and so
  $\PRk{k}{E[e\gamma]} \leq \PR{E'[e\gamma']}$. Let $\vec{\sigma}$ be the types for the right hand side
  in $\phi$. Then $E' \in \Stk{\tau\subst{\Delta, \alpha}{\vec{\sigma}, \sigma'}}$. Using
  the assumption of the lemma we get that $e\gamma' \ciuapproxrel e'\gamma'$ at the type
  $\tau\subst{\Delta, \alpha}{\vec{\sigma}, \sigma'}$ which immediately implies that
  $\PR{E'[e\gamma']} \leq \PR{E'[e'\gamma']}$ concluding the proof.
\end{proof}

\section{The probability of termination}
\label{app:sec:prob-conv}

We prove the claims from Section~\ref{sec:examples} about the termination probability.

\begin{proposition}
  \label{app:prop:left-computable}
  For any expression $e$, $\PR{e}$ is a left-computable real number.
\end{proposition}
\begin{proof}
  We first prove by induction that for any $n$, $\Phi^n(\bot)$ restricts to a map
  $\Exps \to [0,1]\cap \QQ$. The proof is simple since the function $\bot$
  clearly maps into rationals and for the inductive step we use the fact that the sums in
  the definition of $\Phi$ are always finite, and the rational numbers are closed under
  finite sums.

  To conclude the proof we have by definition that $\PR{e} =
  \sup_{n\in\omega}\Phi^n(\bot)(e)$ and we have just shown that all the numbers
  $\Phi^n(\bot)(e)$ are rational. Moreover the sequence
  $\left\{\Phi^n(\bot)(e)\right\}_{n\in\NN}$
  is computable, since for a given $n$ we only need to
  check all the reductions from $e$ of length at most $n$ to determine the value of
  $\Phi^n(\bot)(e)$
  and the reduction relation $\stepsto{p}$ is naturally computable.
\end{proof}

\begin{example}
  \label{app:ex:prob-conv-noncomputable}
  To see that the probability of termination can also be non-computable we informally
  describe a program whose probability of termination would allow us to solve the halting
  problem were it computable.

  The program we construct is recursively defined as $T = \TAPP{\TAPP{\FIX}},\phi$
  where
  \begin{align*}
    \phi = \FUN{f}{\FUN{x}{t\,x \oplus (\Omega \oplus f\,(\syn{succ}~x))}}
  \end{align*}
  where $t\,x$ is a program that runs the $x$-th Turing machine on the empty input and
  does not use any choice reductions. Thus $\PR{t\,x} \in \{0,1\}$. It is well known that
  the empty string acceptance problem is undecidable. Note that we put $\Omega$ in the
  program to ensure that every second digit in binary will be $0$. It is an easy
  computation to show that
  \begin{align*}
    \PR{T\,\underline 1} = \sum_{n=0}^{\infty}\frac{1}{2^{2n+1}}p_{n+1}
  \end{align*}
  where $p_n = 1$ if the $n$-th Turing machine terminates on the empty input and $0$
  otherwise. If $\PR{T\,\underline 1}$ were computable we could decide whether a given
  Turing machine accepts the empty string by computing its index $n$ and then computing
  the first $2n$ digits of $\PR{T\,\underline 1}$.
\end{example}

We will now generalize the last example and show that any left-computable real arises as
the probability of termination of a program. Technically, we show that given a term of the
language that computes an increasing bounded sequence of rationals (represented as pairs
of naturals) we can define a program that terminates with probability the supremum of the
sequence. We then use the fact that our language $\lang$ is Turing complete to claim that
any computable sequence of rationals can be represented as such a term of $\lang$.

\begin{proposition}
  \label{app:prop:every-left-computable-is-termprob}
  For every left-computable real in $[0,1]$ there is a program $e_r$ of type $\ONE\to\ONE$
  such that $\PR{e_r\,\unitval} = r$.
\end{proposition}
\begin{proof}
  So let $r : \NAT \to \NAT \times \NAT$ compute an increasing sequence of rationals in
  the interval $[0,1]$. Additionally assume that for all $n \in \NN$.
  \begin{align*}
    r\,\underline n \stepstopure \PAIR{\underline {k_n}}{\underline {\ell_n}}
  \end{align*}
  for some $k_n, \ell_n \in \NN$. That is, $r$ does not use choice reductions. This is not
  an essential limitation, but simplifies the argument which we are about to give.

  First we define a recursive function $e$ of type $e : (\NAT \to \NAT \times \NAT) \to
  \ONE$ as $e = \TAPP{\TAPP{\FIX}}\,\phi$ where
  \begin{align*}
    \phi = \FUN{f}{\FUN{r}{&\syn{let}~(\underline k,\underline\ell)~=~r\,\underline 1~\syn{in}\\
        &\syn{let}~y~=~\RAND{\underline \ell}~\syn{in}\\
        &\IFT{y \leq \underline k}{\unitval}{f\,r'}}}
  \end{align*}
  and
  \begin{align*}
    r' = \FUN{z}{\frac{r\,(\syn{succ}\,z) - (\underline k, \underline \ell)}{1 -
        (\underline k, \underline \ell)}}
  \end{align*}
  and subtraction and division is implemented in the obvious way. Note that the condition
  in $\phi$ ensures that $(\underline k, \underline \ell)$ does not represent the rational
  number $1$ and therefore division would make sense. But technically, since we implement
  rationals with pairs of naturals no exception can occur and we just represent the pair
  with the second component being $\underline 0$.

  Let $f$ and $r$ be values of the appropriate type. We have
  \begin{align*}
    \PRk{m+1}{\phi\,f\,r} \leq \frac{k_1}{\ell_1}\PRk{m}{\unitval} + \frac{\ell_1 -
      k_1}{\ell_1}\cdot\PRk{m}{f\,r'}
  \end{align*}
  where $r\,\underline 1 \stepstopure (\underline k_1, \underline \ell_1)$. The inequality
  comes from the fact that applying $r$ might take some unfold-fold reductions. Iterating
  this we get
  \begin{align*}
    \PRk{m+1+2n}{e\,r} \leq \frac{k_n}{\ell_n} + \frac{\ell_n -
      k_n}{\ell_n}\cdot\PRk{m+1}{e\,r^{(n)}}
  \end{align*}
  where $r\,\underline n \stepstopure (\underline k_n, \underline \ell_n)$ and
  \[r^{(n)} = \FUN{z}{\frac{r\,(\syn{succ}^n\,z) - (\underline k_n, \underline \ell_n)}{1
      - (\underline k_n, \underline \ell_n)}}\] is the $n$-th iteration of the $'$ used on
  $r$ in $\phi$.

  It is easy to see that $\PRk{1}{e\,r^{(n)}} = 0$ since it takes at least one unfold-fold
  and one choice reduction to terminate.  Thus picking $m=1$ we have $\PRk{2+2n}{e\,r} =
  \frac{k_n}{\ell_n}$ and thus
  \begin{align*}
    \sup_{n\in\omega}\PRk{n}{e\,r} \leq \sup_{n\in\omega}\frac{k_n}{\ell_n}
  \end{align*}

  Using the same reasoning as above we also have
  \begin{align*}
    \PR{e\,r} \geq \frac{k_n}{\ell_n} + \frac{\ell_n - k_n}{\ell_n}\cdot\PR{e\,r^{(n)}}
    \geq \frac{k_n}{\ell_n}
  \end{align*}
  which shows (using Proposition~\ref{prop:prob-smaller}) that
  \begin{align*}
    \sup_{n\in\omega}\frac{k_n}{\ell_n} \leq \PR{e\,r} \leq \sup_{n\in\omega}\PRk{n}{e\,r}
    \leq \sup_{n\in\omega}\frac{k_n}{\ell_n}
  \end{align*}
  and so
  \begin{align*}
    \sup_{n\in\omega}\frac{k_n}{\ell_n} = \PR{e\,r}.
  \end{align*}
\end{proof}

\begin{figure}[htb]
  \centering
  \begin{align*}
    e \oplus e \ctxequivrel e \qquad       e_1 \oplus e_2 \ctxequivrel e_2 \oplus e_1 \qquad
    e \oplus \Omega \ctxapproxrel e\\
    \text{if } e_1 \stepstopure e_2 \text{ then } e_1 \ctxequivrel e_2 \qquad
    \text{if } e_1 \oplus e_2 \ctxequivrel e_1 \text{ then } e_1 \ctxequivrel e_2
  \end{align*}
  \caption{Basic properties of $\ctxapproxrel$ and $\ctxequivrel$. We write $\Omega$ for
    any diverging term (i.e. $\PR{\Omega} = 0$) and $e \oplus e'$ as syntactic sugar for
    $\IFZERO{\RAND{\underline 2}}{e}{e'}$. Note that the choice when evaluating $e \oplus
    e'$ is made \emph{before} $e$ and $e'$ are evaluated.}
  \label{app:fig:basic-properties-ctx}
\end{figure}

\section{Distributions}
\label{app:sec:distributions}

We now define distributions and prove some of their properties and properties of the
probability of termination which are used in the examples.

By a distribution we mean a subprobability measure on the discrete space $\Vals$
of values.
Let
\[
\dists = \{ f : \Vals \to [0,1] | \sum_{v\in\Vals}f(v) \leq 1\}
\]
be the space of subprobability measures on $\Vals$. To be precise, $f \in \dists$ are not
measures, but given any $f$ we can define a subprobability measure $\mu_f(A) = \sum_{v\in
  A}f(v)$ and given any subprobability measure $\mu$, we can define $f_\mu \in \dists$ as
the Radon-Nikodym derivative with respect to the counting measure. Or in more prosaic
terms $f_\mu(v) = \mu\left(\{v\}\right)$. It is easy to see that these two operations are
mutually inverse and since $f\in\dists$ are easier to work with we choose this
presentation.

\begin{lemma}
  \label{app:lem:dists-cpo}
  $\dists$ ordered pointwise is a pointed $\omega$-cpo.
\end{lemma}
\begin{proof}
  The bottom element is the everywhere $0$ function. Let $\{f_n\}_{n\in\omega}$
  be an $\omega$-chain. Define the limit function $f$ as the pointwise supremum
  \begin{align*}
    f(v) = \sup_{n\in\omega}f_n(v).
  \end{align*}
  Clearly all pointwise suprema exist and $f$ is the least upper bound, provided we can
  show that $f\in\dists$. To show this last fact we need to show
  \begin{align*}
    \sum_{v\in\Vals}\sup_{n\in\omega}f_n(v) \leq 1.
  \end{align*}
  but this is a simple consequence of Fatou's lemma since from the assumption that
  $\{f_n\}_{n\in\omega}$ we have
  $\sup_{n\in\omega}f_n(v) = \lim_{n\to\infty}f_n(v) = \liminf_{n\to\infty}f_n(v)$ and so
  by Fatou's lemma (relative to the counting measure on $\Vals$) we have
  \begin{align*}
    \sum_{v\in\Vals}\sup_{n\in\omega}f_n(v) \leq
    \liminf_{n\to\infty}\left(\sum_{v\in\Vals}f_n(v)\right) \leq \liminf_{n\to\infty} 1 = 1.
  \end{align*}
\end{proof}

Now define $\Xi : (\Exps \to \dists) \to (\Exps \to \dists)$ as follows
\begin{align*}
  \Xi(\phi)(e) =
  \begin{cases}
    \delta_e & \text{ if } e\in\Vals\\
    \displaystyle\sum_{e \stepsto{p}{e'}} p \cdot \phi\left(e'\right) &\text{otherwise}
  \end{cases}
\end{align*}
where $\delta_e$ is (the density function of) the Dirac measure at point $e$.  Since
$\dists$ is an $\omega$-cpo so is $\Exps \to \dists$ ordered pointwise. It is easy to see
that in this ordering $\Xi$ is monotone and continuous and so by Kleene's fixed point
theorem it has a least fixed point reached in $\omega$ iterations. Let $\Dl =
\sup_{n\in\omega}\left(\Xi^n(\bot)\right)$ be this fixed point.

\begin{lemma}
  \label{app:lem:dist-nonzero-path}
  Let $e \in \Exps$ and $v \in \Vals$. If $\Dl(e)(v) > 0$ then there exists a path $\pi$
  from $e$ to $v$, i.e. $e$ steps to $v$.
\end{lemma}
\begin{proof}
  We use Scott induction. Define
  \begin{align*}
    \Sl = \left\{ f : \Exps \to \dists \isetsep \forall e, v, f(e)(v) > 0 \implies \exists \pi, \redpath{\pi}{e}{v}\right\}
  \end{align*}
  The set $\Sl$ contains $\bot$. To see that it is closed under $\omega$-chains observe
  that if $\left(\sup_{n\in\omega}f_n\right)(e)(v) > 0$ then there must be $n\in\omega$,
  such that $f_n(e)(v)>0$ so we may use the path from $e$ to $v$ that we know exists from
  the assumption that $f_n \in\Sl$.

  It is similarly easy to see that given $f \in \Sl$ we have $\Xi(f) \in \Sl$. Thus we
  have that $\Dl \in \Sl$ concluding the proof.
\end{proof}

\begin{lemma}
  \label{app:lem:prob-is-sum-of-dist}
  For any expression $e\in\Exps$ we have
  \begin{align*}
    \sum_{v\in\Vals}\Dl(e)(v) = \PR{e}
  \end{align*}
\end{lemma}
\begin{proof}
  First we show by induction on $n$ that all the finite approximations of $\PR{e}$ and
  $\Dl(e)$ agree.
  \begin{itemize}
  \item The base case is trivial since by definition \[\sum_{v\in\Vals}\Xi^0(\bot)(e)(v) = 0 = \Phi^0(\bot)(e)\]
  \item For the inductive case we consider two cases. If $e \in \Vals$ then both sides are
    $1$. In the other case we have
    \begin{align*}
      \sum_{v\in\Vals}\Xi^{n+1}(\bot)(e)(v) &=
      \sum_{v\in\Vals}\left(\sum_{e \stepsto{p}{e'}} p \cdot \Xi^n(e')\right)(v)\\ &=
      \sum_{v\in\Vals}\left(\sum_{e \stepsto{p}{e'}} p \cdot \Xi^n(e')(v)\right)
      \intertext{by Tonelli's theorem we can we can interchange the sums to get}
      &=\sum_{e \stepsto{p}{e'}} \left(p \sum_{v\in\Vals} \Xi^n(e')(v)\right)\\
      &=\sum_{e \stepsto{p}{e'}} p\cdot \Phi^n(\bot)(e') =
        \Phi^{n+1}(\bot)(e)
    \end{align*}
  \end{itemize}
  Thus we have that for all $n$, 
  \begin{align*}
    \sum_{v\in\Vals}\Xi^{n}(\bot)(e)(v) &= \Phi^{n}(\bot)(e)
    \intertext{and so}
    \sup_{n\in\omega}\left(\sum_{v\in\Vals}\Xi^{n}(\bot)(e)(v)\right)
    &= \sup_{n\in\omega}\left(\Phi^{n}(\bot)(e)\right) = \PR{e}
  \end{align*}
  By the dominated convergence theorem we can exchange the $\sup$ (which is the limit) and
  the sum on the left to get
  \begin{align*}
    \sup_{n\in\omega}\left(\sum_{v\in\Vals}\Xi^{n}(\bot)(e)(v)\right) &=
    \sum_{v\in\Vals}\sup_{n\in\omega}\left(\Xi^{n}(\bot)(e)(v)\right) \\ &= \sum_{v\in\Vals}\Dl(e)(v)
  \end{align*}
  as required.
\end{proof}

\begin{proposition}[Monadic bind for distributions]
  \label{app:prop:dist-bind}
  Let $e \in \Exps$ and $E$ an evaluation context of appropriate type.
  \begin{align*}
    \Dl\left(E[e]\right) = \sum_{v\in\Vals}\Dl(e)(v)\cdot\Dl\left(E[v]\right).
  \end{align*}
\end{proposition}
\begin{proof}
  It is easy to show by induction on $\ell$ that
  \begin{align}
    \label{app:eq:bind-approx}
    \forall e \in \Exps, \Xi^\ell\left(\bot\right)\left(E[e]\right) =
    \sum_{v\in\Vals}\sum_{\substack{\redpath{\pi}{e}{v}\\\len{\pi}\leq \ell}}\weight{\pi}\cdot\Xi^{\ell-\len{\pi}}\left(E[v]\right)
  \end{align}
  (using the fact that the length of the empty path is $0$ and its weight $1$).

  Similarly it is easy to show by induction on $\ell$ that
  \begin{align}
    \label{app:eq:approx-distribution-paths}
    \forall e \in \Exps, \Xi^{\ell+1}\left(\bot\right)\left(e\right)(v) &=
    \sum_{\substack{\redpath{\pi}{e}{v}\\\len{\pi}\leq \ell}}\weight{\pi}
  \end{align}
  which immediately implies
  \begin{align}
    \forall e \in \Exps, \Dl(e)(v) &=
    \sum_{\redpath{\pi}{e}{v}}\weight{\pi}
  \end{align}
  Using these we have
  \begin{align*}
    \Dl(E[e]) &= \sup_{\ell\in\omega}
                 \sum_{v\in\Vals}\sum_{\substack{\redpath{\pi}{e}{v}\\\len{\pi}\leq
    \ell}}\weight{\pi}\cdot\Xi^{\ell-\len{\pi}}\left(E[v]\right)
    \intertext{and since for each $v$ the sequence $\displaystyle\sum_{\substack{\redpath{\pi}{e}{v}\\\len{\pi}\leq
    \ell}}\weight{\pi}\cdot\Xi^{\ell-\len{\pi}}\left(E[v]\right)$ is increasing with
    $\ell$ we have}
    &= \sum_{v\in\Vals}
                 \sup_{\ell\in\omega}\sum_{\substack{\redpath{\pi}{e}{v}\\\len{\pi}\leq
    \ell}}\weight{\pi}\cdot\Xi^{\ell-\len{\pi}}\left(E[v]\right)\\
              &= \sum_{v\in\Vals} \sum_{\redpath{\pi}{e}{v}}\weight{\pi}\cdot\Dl\left(E[v]\right)\\
              &= \sum_{v\in\Vals}
                \Dl\left(E[v]\right)\sum_{\redpath{\pi}{e}{v}}\weight{\pi}\\
              &= \sum_{v\in\Vals}\Dl(e)(v)\cdot\Dl\left(E[v]\right) 
  \end{align*}
\end{proof}
\begin{corollary}
  \label{app:cor:dist-bind}
  Let $e \in \Exps$ be typeable and $E$ an evaluation context of appropriate type. Then
  $\PR{E[e]} = \sum_{\redpath{\pi}{e}{v}}\Wl(\pi)\cdot\PR{E[v]}$.
\end{corollary}

\begin{corollary}
  \label{app:cor:pr-and-dists}
  For any term $e$ and evaluation context $E$ the equality
  \begin{align*}
    \PR{E[e]} = \sum_{v\in\Vals}\Dl(e)(v)\cdot\PR{E[v]}
  \end{align*}
  holds.
\end{corollary}

\begin{corollary}
  \label{app:cor:point-mass-bind-prob}
  Let $e \in \Exps$ and $E$ an evaluation context. Suppose $\Dl(e) = p\cdot\delta_v$ for
  some $v\in\Vals$ and $p\in [0,1]$. Then $\PR{E[e]} = p\cdot\PR{E[v]}$.
\end{corollary}
\begin{proof}
  Use Proposition~\ref{app:prop:dist-bind} and Lemma~\ref{app:lem:prob-is-sum-of-dist}.
\end{proof}

\begin{proposition}
  \label{app:prop:prk-eval-ctxts}
  For any evaluation context $E$ and term $e$ and any $k\in\NN$, 
  \begin{align*}
    \PRk{k}{E[e]} \leq \sum_{\redpath{\pi}{e}{v}}\weight{\pi}\cdot\PRk{k}{E[v]}
  \end{align*}
\end{proposition}
The proof proceeds by induction on $k$. 

\section{Further examples}
In this section we show further equivalences which did not fit into the paper proper due
to space restrictions.

\subsubsection{Fair coin from an unfair one}
\label{app:sec:fair-coin-from}

Given an unfair coin, that is, a coin that comes up heads with probability $p$ and tails
with probability $1-p$, where $0 < p < 1$ we can derive an infinite sequence of fair coin
tosses using the procedure proposed by von Neumann. The procedure follows from the
observation that if we toss an unfair coin twice, the likelihood of getting (H, T) is the
same as the likelihood of getting (T, H). So the procedure works as follows

\begin{itemize}
\item Toss the coin twice
\item If the result is (H, T) or (T, H) return the result of the first toss
\item Else repeat the process
\end{itemize}

We only consider rational $p$ in this section (for a computable $p$ we could proceed
similarly, but the details would be more involved, since the function which returns $1$
with probability $p$ and $0$ with probability $1-p$ is a bit more challenging to write).

Let $1 \leq k < n$ be two natural numbers and $p = \frac{k}{n}$. Below we define $e_p :
\ONE \to \BOOL$ to be the term implementing the von Neumann procedure for generating fair
coin tosses from an unfair coin $t_p$ which returns $\TRUE$ with probability $p$ and
$\FALSE$ with proability $1-p$.  We will show that $e_p$ is contextually equivalent to
$\FUN{x}{\TRUE \oplus \FALSE}$.
We define $e_p$ as 
\begin{align*}
  e_p = \TAPP{\TAPP{\FIX}} \phi
\end{align*}
where
\begin{align*}
  \BOOL &= \ONE + \ONE\\
  \TRUE &= \INL{\unitval}\\
  \FALSE &= \INR{\unitval}\\
  e \equiv e' &= \MATCH{e}{\_}{e'}{\_}{\MATCH{e'}{\_}{\FALSE}{\_}{\TRUE}}\\
  \IFT{e}{e_1}{e_2} &= \MATCH{e}{\_}{e_1}{\_}{e_2}\\
  t_p &= \FUN{\unitval}{\syn{let}~y~=\RAND{\underline n}~\syn{in}~(y \leq \underline k)}
\end{align*}
and
\begin{align*}
  \phi = \FUN{f}{\FUN{\unitval}{&\syn{let}~x~=~t_p\, \unitval~\syn{in}\\
                                &\syn{let}~y~=~t_p\,\unitval~\syn{in}\\
                                &\IFT{x \equiv y}{f\,\unitval}{x}}}.
\end{align*}

By a simple calculation using the operational semantics we can see that given any
evaluation context $E$, we have
$\PR{E[t_p\,\unitval]} = \frac{k}{n}\PR{E[\TRUE]} + \frac{n-k}{n}\PR{E[\FALSE]}$.
Given any value $f$ of type ($\ONE \to \BOOL$) and any evaluation context $E$
with the hole of type $\BOOL$ we compute that $\PR{E[\phi\,f\,\unitval]}$ is equal to
$\frac{k^2+(n-k)^2}{n^2}\PR{E[f\,\unitval]} + 2\cdot\frac{k\cdot(n-k)}{n^2}\PR{E[\TRUE \oplus \FALSE]}$.
Finally for $e_p$ and any evaluation context $E$ with hole of type $\BOOL$ we have
\begin{align*}
  \PR{E[e_p\,\unitval]} &= \PR{\phi\,e_p\,\unitval}
  = \frac{k^2+(n-k)^2}{n^2}\PR{E[e_p\,\unitval]}\\ &+
                          2\cdot\frac{k\cdot(n-k)}{n^2}\PR{E[\TRUE \oplus \FALSE]}.
\end{align*}
from which we have by simple algebraic manipulation that
$\PR{E[e_p\,\unitval]} = \PR{E[\TRUE\oplus\FALSE]}$.

It is now straightforward to show
${\logequiv{\emp}{\emp}{e_p}{\FUN{\unitval}{\TRUE\oplus\FALSE}}{\ONE\to\BOOL}}$
since both $e_p$ and $\FUN{\unitval}{\TRUE\oplus\FALSE}$ are values, so we can show them
related in the value relation. The proof uses reflexivity of $\logequivrel$.

Alternatively, we could have used Theorem~\ref{thm:CIU-theorem} and showed directly that
$e_p\,\unitval$ and $\TRUE\oplus\FALSE$ are CIU-equivalent and then used extensionality
for values to conclude the proof.

\subsubsection{A hesitant identity function}
\label{app:sec:hesit-ident-funct}

We consider the identity function $e$ that does not return
immediately, but instead when applied to a value $v$ flips a coin whether to return $v$ or
call itself recursively with the same argument. We show that this function is contextually
equivalent to the identity function $\FUN{x}{x}$. The reason for this is, intuitively,
that even though $e$ when applied may diverge, the probability of it doing so is $0$.

\begin{example}
  \label{app:ex:hesitant-identity}
  Let $e = \TAPP{\TAPP{\FIX}} \left(\FUN{f}{\FUN{x}{(x \oplus f\,x)}}\right) : \alpha \to
  \alpha$.
  We have
  \[\logapprox{\alpha}{\emp}{e}{\FUN{x}{x}}{\alpha\to\alpha}
  \]
  and 
  \[\logapprox{\alpha}{\emp}{\FUN{x}{x}}{e}{\alpha\to\alpha}.\]
\end{example}
\begin{proof}
  We prove the two approximations separately. Let $\phi \in \VRelD{\alpha}$,
  $n\in\NN$. Since $e$ and $\FUN{x}{x}$ are values we show them directly related in the
  value relation. In both cases let $\phi = \FUN{f}{\FUN{x}{(x \oplus f\,x)}}$ and $h =
  \FUN{z}{\delta_\phi\,(\FOLD{\delta_\phi})\,z}$.
  \begin{itemize}
  \item By definition of the
    interpretation of function types we have to show, given $k\leq n$ and $(v, v') \in
    \phi_r(\alpha)(k)$, that $(e\,v, (\FUN{x}{x})\,v') \in \TT{\phi_r(\alpha)}(k)$.

    It is straightforward to see that $e\,v \stepstopure \phi\,e\,v$ using exactly one
    unfold-fold reduction.

    Now let $(E, E')$ be related at $k$.
    We proceed by induction and show that for every $\ell \leq k$, 
    $\PRk{\ell}{E[e\,v]}\leq\PR{E'[v']}$ which suffices by
    Lemma~\ref{lem:choice-free-reductions-dont-change}. 
    When $\ell = 0$ there is nothing to prove.
    So let $\ell = \ell' + 1$.
    \begin{align*}
      \PRk{\ell}{E[e\,v]} = \PRk{\ell'}{\phi\,e\,v} = \PRk{\ell'}{E[v \oplus e\,v]}.
    \end{align*}
    If $\ell' = 0$ we are trivially done. So suppose $\ell' = \ell'' + 1$ to get using
    Lemma~\ref{lem:prob-unfold-fold-choice} 
    \begin{align*}
      \PRk{\ell'}{E[v \oplus e\,v]} = \frac{1}{2}\PRk{\ell''}{E[v]} + \frac{1}{2}\PRk{\ell''}{e\,v}
    \end{align*}
    Using the fact that $\ell'' \leq k$ and monotonicity we have 
    \[\PRk{\ell''}{E[v]} \leq \PR{E'[v']}.\]
    Using the induction hypothesis we have
    \[\PRk{\ell''}{e\,v} \leq \PR{E'[v']}\]
    which together conclude the proof.
  \item Again by definition of the interpretation of function types we have to show, given
    $k\leq n$ and $(v, v') \in \phi_r(\alpha)(k)$, that $\left((\FUN{x}{x})\,v', e\,v\right)  \in
    \TT{\phi_r(\alpha)}(k)$.

    Again we have that $e\,v' \stepstopure \phi\,e\,v'$ using exactly one unfold-fold
    reduction.  Let $\ell \leq k$ and $(E, E')$ related at $\ell$. Using
    Lemma~\ref{lem:choice-free-reductions-dont-change} and the fact that $\PR{\cdot}$ is a
    fixed point of $\Phi$ we have
    \begin{align*}
      \PR{E'[e\,v']} &= \PR{E'[\phi\,e\,v']}\\
      &= \frac{1}{2}\PR{E'[v']} + \frac{1}{2}\PR{E'[e\,v']}
    \end{align*}
    and from this we get $\frac{1}{2}\PR{E'[e\,v']} = \frac{1}{2}\PR{E'[v']}$
    by simple algebraic manipulation and thus $\PR{E'[e\,v']} = \PR{E'[v']}$. Using this property
    it is a triviality to finish the proof.
  \end{itemize}
\end{proof}

\subsection{Further simple examples}
\label{app:sec:exampl-from-bisim}

The following example is a proof of \emph{perfect security} for the one-time pad
encryption scheme. Define the following functions

\begin{align*}
  \syn{not} &: \BOOL \to \BOOL\\
  \syn{not} &= \FUN{x}{\IFT{x}{\FALSE}{\TRUE}}\\
  \syn{xor} &: \BOOL\to\BOOL\to\BOOL\\
  \syn{xor} &= \FUN{x}{\FUN{y}{\IFT{x}{\syn{not}~y}{y}}}\\
  \syn{gen} &: \BOOL\\
  \syn{gen} &= \TRUE \oplus \FALSE
\end{align*}
$\syn{xor}$ is supposed to be the encryption function, with the first argument the
plaintext and the second one the encryption key.

We now encode a game with two players. The first player chooses two plaintexts and gives
them to the second player, who encrypts one of them (using $\syn{xor}$) chosen at random
with uniform probability and gives the result back to the first player. The first player
should not be able to guess which of the plaintexts was encrypted. This is expressed as
contextual equivalence of the following two programs
\begin{align*}
  \syn{exp} &= \FUN{x}{\FUN{y}{\syn{xor}~(x \oplus y)~\syn{gen}}}\\
  \syn{rnd} &= \FUN{x}{\FUN{y}{\syn{gen}}}
\end{align*}

To show $\syn{exp} \ctxequivrel \syn{rnd}$ we first use extensionality for values so we
only need to show that for all $v, u \in \Val{\BOOL}$
\begin{align*}
  \syn{xor}~(v \oplus u)~\syn{gen} \ctxequivrel \syn{gen}
\end{align*}
and the easiest way to do this is by using CIU equivalence. Given an evaluation context
$E$ we have
\begin{align*}
  \PR{E[\syn{xor}~(v \oplus u)~\syn{gen}]} = 
   \frac{1}{4}\left(\begin{array}{l}\PR{E[\syn{xor}~v~\TRUE]} +\\
                      \PR{E[\syn{xor}~v~\FALSE]} +\\
                      \PR{E[\syn{xor}~u~\TRUE]} +\\
                      \PR{E[\syn{xor}~u~\FALSE]}\end{array}\right)
\end{align*}
and by the canonical forms lemma $u$ and $v$ can be either $\TRUE$ or $\FALSE$. It is easy
to see that the sum evaluates to 
\begin{align*}
  \frac{1}{4}(2\cdot\PR{E[\TRUE]} + 2\cdot\PR{E[\FALSE]})
\end{align*}
quickly leading to the desired conclusion.

If we had used the logical relation directly we would not need the canonical forms lemma,
but then we would have to take care of step-indexing.

A similar example is when in one instance we choose to encrypt the first plaintext and in
the second instance the second one. Since the key is generated uniformly at random,
the first player should not be able to distinguish those two
instances. Concretely, this is expressed as contextual equivalence of the following two programs
\begin{align*}
  \syn{exp}_1 &= \FUN{x}{\FUN{y}{\syn{xor}~x~\syn{gen}}}\\
  \syn{exp}_2 &= \FUN{x}{\FUN{y}{\syn{xor}~y~\syn{gen}}}
\end{align*}

The proof is basically the same as the one above. Use extensionality and then CIU
equivalence.

\subsection{Restrictions in the free theorem are necessary}
\label{app:sec:restr-necessary}

We show that the free theorem in Section~\ref{sec:free-theorem} does not hold
without the assumptions on the behaviour of functions $f$ and $g$.

First, if $f = (\FUN{x}{\underline 1}) \oplus (\FUN{x}{\underline 2})$, $g$ is the identity function
$\FUN{x}{x}$ and $xs$ is the list $[\unitval, \unitval]$ then
the term $\TAPP{\TAPP{\syn{map}}}(f \comp g)\,xs$ can reduce to the list $[\underline
1,\underline 2]$, however 
the term $((\TAPP{\TAPP{\syn{map}}}\,f) \comp (\TAPP{\TAPP{\syn{map}}}\,g))\,xs$ cannot.
The reason is that in the first case the reduction of $f$ is performed for each element of
the list separately, but in the latter case, $f$ is first reduced to a value and then the
same value is applied to all the elements of the list. Technically, the condition we need
for $f$ is that there exists a value $f'$, such that $f \ctxequivrel f'$, but this version is
easily derived from the version stated above by congruence.

Second, if $g$ diverges with a non-zero probability for some value $v$, we take $m$ to be
the constant function returning the empty list and the list $xs$ to be the singleton list
containing only the value $v$. Then, if $f$ is any value,
$\TAPP{\TAPP{m}}\,(f \comp g)\,xs$ reduces to the empty list with probability $1$, however
$((\TAPP{\TAPP{m}}{f} \comp \TAPP{\TAPP{\syn{map}}}\,g))\,xs$ reduces to the empty list with
a probability smaller than $1$, since $g$ is still applied, since we are in a
\emph{call-by-value} language.

Third, if $g = \FUN{x}{\underline 1 \oplus \underline 2}$, $f$ is the identity function
and $xs$ is the singleton list containing $\unitval$ we take $m$ to be the function that
first appends the given list to itself and \emph{then} applies $\syn{map}$ to it. We then
have that $\TAPP{\TAPP{m}}\,(f \comp g)\,xs$ \emph{can} reduce to the list $[\underline
1,\underline 2]$, but $((\TAPP{\TAPP{m}}{f}) \comp (\TAPP{\TAPP{\syn{map}}}\,g))\,xs$ cannot,
since $g$ is only mapped over the singleton list producing lists $[\underline 1]$ and
$[\underline 2]$, which are then appended to themselves, giving lists $[\underline 1,
\underline 1]$ and $[\underline 2,\underline 2]$.

And last, if $m$ is not equivalent to a term of the form $\TFUN{\TFUN{\FUN{x}{e}}}$ then
the term on the left reduces to two different (not equivalent) values (or even diverges),
but the term on the right does not. We can use this to construct a distinguishing
evaluation context.

\subsection{A property of $\syn{map}$}
\label{app:sec:property-synmap}

The result in Section~\ref{sec:free-theorem} does not allow us to conclude
\begin{align*}
  \TAPP{\TAPP{\syn{map}}}\,(f \comp g) \ctxequivrel \TAPP{\TAPP{\syn{map}}}\,f\comp \TAPP{\TAPP{\syn{map}}}\,g.
\end{align*}
for all $f \in \Val{\sigma \to \rho}$ and $g\in \Val{\tau\to\sigma}$, however we can show,
using the definition of $\syn{map}$, that this does in fact hold. By using
extensionality (Lemma~\ref{lem:func-extensionality}) we need to show for any list $xs$ we
have
\begin{align*}
  \TAPP{\TAPP{\syn{map}}}\,(f \comp g)\,xs \ctxequivrel \left(\TAPP{\TAPP{\syn{map}}}\,f\comp \TAPP{\TAPP{\syn{map}}}\,g\right)\,xs.
\end{align*}

If $f$ and $g$ are values, $E$ an evaluation context and $xs$ a list of length $n$, it is
easy to see that
\begin{align*}
  \PR{E[\syn{map}\,f\,xs]} &= \sum_{us}\left(\prod_{i=1}^n\Dl(f\,x_i)(u_i)\right)\cdot\PR{E[us]}
\end{align*}
where the first sum is over all the lists of length $n$ and $x_i$ and $u_i$ are the $i$-th
elements of lists $xs$ and $us$, respectively. This then gives us
that
\[\PR{E[\syn{map}\,f\,(\syn{map}\,g\,xs)]}\]
is equal to
\begin{align*}
  &\sum_{vs}\left(\prod_{i=1}^n\Dl(g\,x_i)(v_i)\right)\cdot\PR{E[\syn{map}\,f\,vs]}\\
  &= \sum_{vs}\left(\prod_{i=1}^n\Dl(g\,x_i)(v_i)\right)\cdot\left(\sum_{us}\left(\prod_{i=1}^n\Dl(f\,v_i)(u_i)\right)\cdot\PR{E[us]}\right)\\
  &= \sum_{vs}\sum_{us}\left(\prod_{i=1}^n\Dl(g\,x_i)(v_i)\cdot\Dl(f\,v_i)(u_i)\right)\cdot\PR{E[us]}.
\end{align*}

On the other hand, we have that $\PR{E[\syn{map}\,(f\comp g)\,xs]}$ is equal to
\begin{align*}
  \sum_{us}\left(\prod_{i=1}^n\Dl((f \comp g)\,x_i)(u_i)\right)\cdot\PR{E[us]}
\end{align*}
and 
\begin{align*}
  \Dl((f \comp g)\,x_i)(u_i) = \sum_{v} \Dl(g\,x_i)(v)\cdot \Dl(f\,v)(u_i)
\end{align*}
together giving us 
\begin{align*}
  \sum_{us}\left(\prod_{i=1}^n\left(\sum_{v} \Dl(g\,x_i)(v)\cdot \Dl(f\,v)(u_i)\right)\right)\cdot\PR{E[us]}
\end{align*}
which by Fubini's theorem and the fact that lists of length $n$ correspond to $n$-tuples, is equal to
\begin{align*}
  \sum_{us}\sum_{vs}\left(\prod_{i=1}^n\left(\Dl(g\,x_i)(v_i)\cdot \Dl(f\,v_i)(u_i)\right)\right)\cdot\PR{E[us]}
\end{align*}
which is the same as $\PR{E[\syn{map}\,f\,(\syn{map}\,g\,xs)]}$.

If $f$ and $g$ are not equivalent to values, then the above result for $\syn{map}$ does not
hold. Consider, for instance, $f = \FUN{x}{\underline 1} \oplus \FUN{x}{\underline 2}$
and $g$ the identity or conversely, when applied to the list $xs = [\unitval,
\unitval]$. The expression $\TAPP{\TAPP{\syn{map}}}\,(f \comp g)\,xs$ can reduce to the
list $[1,2]$, whereas the expression $(\TAPP{\TAPP{\syn{map}}}\,f\comp
\TAPP{\TAPP{\syn{map}}}\,g)\,xs$ cannot. We can generalize this to show that if $f$ is
not equivalent to a value or $g$ is not, then the stated equality does not hold.

\end{document}